\documentclass[a4paper]{article}

\usepackage{a4wide}
\usepackage{amsmath,amsfonts,amssymb,amsthm,dsfont}

\usepackage[active]{srcltx}

\usepackage{mathrsfs}
\usepackage{amscd}

\newcounter{theorem}
\renewcommand{\thetheorem}{\arabic{section}.\arabic{theorem}}

\newenvironment{thm}[1]{\par
\begin{sloppypar}\refstepcounter{theorem}%
\noindent{\bf #1 \thetheorem.}\it{}}{\end{sloppypar}}

\newenvironment{theorem}{\begin{thm}{Theorem}}{\end{thm}}
\newenvironment{proposition}{\begin{thm}{Proposition}}{\end{thm}}
\newenvironment{corollary}{\begin{thm}{Corollary}}{\end{thm}}
\newenvironment{lemma}{\begin{thm}{Lemma}}{\end{thm}}

\newenvironment{defi}[1]{\par
\begin{sloppypar}\refstepcounter{theorem}%
\noindent{\bf #1 \thetheorem.}\rm{}}{\end{sloppypar}}
\newenvironment{definition}{\begin{defi}{Definition}}{\end{defi}}

\newenvironment{remark}{\begin{defi}{Remark}}{\end{defi}}
\newenvironment{hypothesis}{\begin{defi}{Hypothesis}}{\end{defi}}

\def\R{{\rm I\kern-.2em R}}

 \def\W{{\cal W}}   

\def\X{\mathcal X}
\def\V{\mathcal{V}}

\def\s0{\sigma_0}

\def\z{\mathfrak{z}}

\def\Rd{\mathbb{R}^d}
\def\bb1{{\rm{1}\hspace{-3pt}\mathbf{l}}}
\def\Ie0{[-\epsilon_0,\epsilon_0]}

\def\beq{\begin{equation}}
\def\eeq{\end{equation}}

\numberwithin{equation}{section}
\numberwithin{theorem}{section}
\title{Peierls substitution and magnetic pseudo-differential calculus}

\begin{document}

\maketitle

\begin{center}

\small{
Horia D. Cornean\footnote{Department of Mathematical Sciences,     Aalborg University, Fredrik Bajers Vej 7G, 9220 Aalborg, Denmark}, Viorel Iftimie\footnote{Institute of Mathematics of the Romanian Academy,
P.O. Box 1-764, RO-70700, Bucharest, Romania}, Radu Purice\footnote{Institute of Mathematics of the Romanian Academy, Research unit No. 3, 
P.O. Box 1-764, RO-70700, Bucharest, Romania, and
Laboratoire Europ{\'e}en Associ{\'e} CNRS Franco-Roumain Math-Mode
Bucharest
Romania}}

\end{center}

\vspace{0.5cm}

\begin{abstract}
We revisit the celebrated Peierls-Onsager substitution employing the magnetic pseudo-differential calculus for weak magnetic fields with no spatial decay conditions, when the non-magnetic symbols have a certain spatial periodicity. We show in great generality that the symbol of the magnetic band Hamiltonian admits a convergent expansion. Moreover, if the non-magnetic band Hamiltonian admits a localized composite Wannier basis, we show that the magnetic band Hamiltonian is unitarily equivalent to a Hofstadter-like magnetic matrix. In addition, if the magnetic field perturbation is slowly variable, then 
the spectrum of this matrix is close to the spectrum of a Weyl quantized, minimally coupled symbol.   
\end{abstract}

\tableofcontents

\section{Introduction and main results}

The mathematical description of the quantum theory of solids
is mainly based on the spectral and dynamical analysis of Schr\"odinger-type operators with
periodic coefficients. A central problem is concerned with the response of a periodic system when submitted to an external
magnetic field which does not
vanish at infinity. The main technical difficulty in dealing with such a problem comes
from the fact that the
perturbation represented by the magnetic field is no longer an analytic
perturbation. In particular, the
difference between the perturbed and the unperturbed observables is not 
small in operator norm for small values of the magnetic field intensity and a norm convergent 
perturbative approach is impossible.

The solution proposed by physicists (see section 7.3
in \cite{Pe}), the so-called {\it
Peierls-Onsager substitution}, is to isolate as main contribution the
magnetic
quantization of the Bloch eigenvalues and, modulo some gauge transformation, to
obtain a power series development with respect to the magnetic field intensity. In that framework, {\it magnetic quantization} means just to replace the
canonical moment variables with the magnetic ones by the so called {\it minimal
coupling procedure}. In mathematical terms this amounts to give a precise meaning to functions of the operator $\theta-A(i\nabla_\theta)$
for $\theta\in\mathbb{T}_*$ the $d$-dimensional torus in the momentum space (see
\cite{FrTe} and references therein). Moreover, one has to to
deal with an apparent lack of gauge covariance of this procedure. Although
a rather rich literature is devoted to the precise mathematical treatment of
this problem (see \cite{B, Bus, dNL,FrTe,  GRT,HS, IP14,Ne-RMP, PST}), a complete understanding is still
missing. 

The main objective of our paper is to show that the {\it magnetic pseudodifferential
calculus} developed in \cite{MP1,IMP1,IMP2} and its {\it integral kernel}
counterpart developed in \cite{Ne05} can give a precise mathematical framework and 
a gauge covariant formulation of the above physical proposal by -roughly speaking- replacing $E_n(k-A(x))$ of the Peierls-Onsager substitution method (here $E_n$ is the $n$-th Bloch eigenvalue of the unperturbed system) with the magnetic pseudodifferential operator 
$\mathfrak{Op}^A(E_n)$.
In this paper we shall mainly concentrate on the so-called {\it isolated
spectral island} situation (see Hypothesis \ref{is-sp-band}) and leave the case without a spectral gap for a
forthcoming paper.

Let us point out from the very beginning that our treatment is gauge
covariant. The magnetic fields are supposed to be smooth and bounded but no
condition of slow variation is imposed in general. Of course, under an extra hypothesis of slow variation, stronger results may be obtained and we shall compare them with those already existing in the literature. Let us also underline that the results
obtained cover a large class of Hamiltonians described by
periodic pseudodifferential
operators.

One of the conclusions of our analysis is that
once we can isolate a spectral island of a periodic Hamiltonian and define
symbols of the associated band operators
(Hamiltonian, projection, etc.), then their corresponding magnetic counterparts, while
being 
singular perturbations of the free ones, can be well approximated in the norm topology by the magnetic
quantization of the `free' symbols.  Moreover, when a composite Wannier basis is supposed to exist
for the unperturbed spectral projection \cite{CHN, FMP}, then a generalization of the results in \cite{Ne-RMP} is
obtained.

\subsection{The framework and some notation}

We shall consider a $d$-dimensional configuration space with $d\geq 2$ and use the
notation $\X\equiv\Rd$. Although its dual is
canonically isomorphic to $\Rd$ we shall prefer to use the notation $\X^*$
in order to emphasize the dual variable. We shall denote by
$\langle \cdot ,\cdot \rangle :\X^*\times\X\rightarrow\mathbb{R}$
the duality relation. We shall consider $\Xi:=\X\times\X^*$ as a symplectic space with the canonical symplectic form
$\sigma^\circ(X,Y):=\langle\xi,y\rangle-\langle\eta,x\rangle$ where $X=(x,\xi)$ and $Y=(y,\eta)$.

We shall consider an algebraic basis $\{e_1,\ldots, e_d\}$ of $\Rd$ and the 
lattice $\Gamma=\oplus_{j=1}^d\mathbb{Z}e_j$; it gives an injective
homomorphism $\mathbb{Z}^d\subset\X$ that allows us to view $\Gamma$ as a
discrete subgroup of $\X$.
We consider the quotient group $\X/\Gamma$ that is canonically isomorphic to
the $d$-dimensional torus $\mathbb{T}^d$ that we shall
denote by  $\mathbb{T}$ when the subgroup
$\Gamma$ is evident; let us denote by $\pi:\X\rightarrow\mathbb{T}$ the
canonical projection onto the quotient.
Let us fix an \textit{elementary cell}:
$$
E\,:=\,\left\{y=\sum\limits_{j=1}^dt_je_j\in\Rd\,\mid\,
-1/2\leq t_j<1/2,\
\forall j\in\{1,\ldots,d\}\right\},
$$
and the
unique decomposition it induces for any $x\in\X$: $x=[x]+\hat{x}$ with
$[x]\in\Gamma$ and $\hat{x}\in E$.
The dual lattice of $\Gamma$ is 
defined as
$$
\Gamma_*\,:=\,\left\{\gamma^*\in\X^*\,\mid\,\langle\gamma^*,\gamma\rangle /(2\pi)\in \mathbb{Z}
,\ \forall\gamma\in\Gamma\right\}.
$$
Considering the dual basis $\{e^*_1,\ldots,e^*_d\}\subset\X^*$ defined by
$\langle e^*_j,e_k\rangle =(2\pi)\delta_{jk}$, we have $\Gamma_*=\oplus_{j=1}^d\mathbb{Z}e^*_j$. We define
$\mathbb{T}_*:=\X^*/\Gamma_*$ and
$E_*$ by a definition
similar to that of $E$ but with respect to the dual basis $\{e^*_j\}_{1\leq
j\leq d}$. We notice that
$\mathbb{T}_*$ is isomorphic to the dual group of $\Gamma$ and $\Gamma_*$ is
isomorphic to the dual group of $\mathbb{T}$ (in the
sense of abelian locally compact groups); moreover
let ${\pi}_*:\X^*\rightarrow\mathbb{T}_*$ be
the canonical quotient projection. 
For any $\gamma\in\Gamma$, $\gamma^*\in\Gamma_*$ or $x\in\X$, and $\xi\in\X^*$
we shall denote by $\sigma_{\gamma}$, $\sigma_{\gamma^*}$, $\sigma_x$ and
$\sigma_\xi$ the
corresponding characters on $\mathbb{T}_*$, $\mathbb{T}$,  $\X$ and $\X^*$ respectively,
i.e.
$$
\sigma_\gamma(\theta):=e^{-i\langle\theta,\gamma\rangle },\quad\sigma_{\gamma^*}(\pi(x)):=e^
{-i\langle\gamma^*,\pi(x)\rangle},\quad\sigma_x(\xi)=\sigma_\xi(x):=e^{-i\langle\xi,x\rangle}.
$$

Given some finite dimensional real Euclidean vector space $\mathcal{V}$, we
shall denote by $\mathscr{S}(\mathcal{V})$ and $\mathscr{S}^\prime(\mathcal{V})$
the Schwartz space of test functions and resp. its dual, the space of tempered
distributions, on $\mathcal{V}$; we shall denote by
$\langle \cdot ,\cdot \rangle_{\mathcal{V}}:\mathscr{S}^\prime(\V)\times\mathscr{S}
(\V)\rightarrow\mathbb
{C}$ the canonical duality. We shall constantly use the notation
$< v> :=\sqrt{1+|v|^2}$ for any $v\in\V$. 
We shall consider the spaces $BC(\V)$ of bounded
continuous functions with the
$\|\cdot\|_\infty$ norm.
We shall denote by $C^\infty(\mathcal{V})$ the space of smooth functions on $\mathcal{V}$ and by $C^\infty_{\text{\sf
pol}}(\mathcal{V})$ and by $BC^\infty(\V)$ its subspace of smooth functions
that are polynomially bounded together with all their derivatives or smooth and bounded
together with all
their derivatives.

Let us recall a class of H\"{o}rmander type symbols on $\Xi$ that we shall use.
For any
$s\in\mathbb{R}$ and any $\rho\in[0,1]$ we denote by
\beq
S^s_\rho(\X):=\{F\in
C^\infty(\Xi)\mid\nu^{s,\rho}_{a,b}(F)<\infty,\forall(a,b)\in\mathbb{N}
^d\times\mathbb{
N}^d\}
\eeq 
where
$\nu^{s,\rho}_{a,b}(f):=\underset{(x,\xi)\in\Xi}{\sup}\left|\langle \xi\rangle ^{-s+\rho|b|}
\big(\partial^a_x\partial^b_\xi f\big)(x,\xi)\right|$, 
$\forall(a,b)\in\mathbb{N}^d\times\mathbb{N}^d$; we recall that $|b|:=\sum\limits_{j=1}^{d}b_j$. A symbol $F$ in $S^s_\rho(\X)$ is called {\it elliptic}  
if there exist two positive constants $R$ and $C$ such that $|F(x,\xi)|\geq C\langle \xi\rangle ^s$ for any
$(x,\xi)\in\Xi$ with $|\xi|\geq R$. 

Let us recall that for
$h\in S^m_1(\X)$ the \textit{Weyl quantization} associates the operator
$\mathfrak{Op}(h)$ defined on $\mathscr{S}(\X)$ (see formula \eqref{OpA} with $A=0$), and having a natural extension by duality to
$\mathscr{S}^\prime(\X)$.

We denote by $\tau$ the action by translations both of $\X$ and of $\X^*$
on
the tempered distributions defined on the respective spaces, i.e. $\big(\tau_zF\big)(x):=F(x+z)$ for any $F\in C(\X)$, and any $(z,x)\in\X\times\X$. We shall consider  symbols of H\"{o}rmander
type $S^m_{1}(\X)$ satisfying
$\big((\tau_\gamma\otimes\bb1)h\big)(x,\xi)=h(x,\xi)$ for
any $\gamma\in\Gamma$ and for any $(x,\xi)\in\Xi$
and we denote by $S^m_1(\X)_\Gamma$ the space of these symbols.
We shall use the following notations:
\begin{itemize}
\item $\mathscr{S}^\prime_\Gamma(\X)\ :=\
\left\{u\in\mathscr{S}^\prime(\X)\,\mid\,\tau_\gamma u=u,\
\forall\gamma\in\Gamma\right\}$, the space of $\Gamma$-periodic distributions on
$\X$.
\item $\mathscr{S}(\mathbb{T}):=C^\infty(\mathbb{T})$ with the usual Fr\'{e}chet
topology; $\mathscr{S}^\prime(\mathbb{T})$ is the dual of
$\mathscr{S}(\mathbb{T})$. We shall denote by
$\langle\cdot,\cdot\rangle_{\mathbb{T}}$ the natural bilinear map
defined by the duality relation on
$\mathscr{S}^\prime(\mathbb{T})\times\mathscr{S}(\mathbb{T})$.
\end{itemize}
\begin{remark}\label{R.A.9}
We have the identification
$
\mathscr{S}(\mathbb{T})\ \cong\
\mathscr{S}^\prime_\Gamma(\X)\cap C^\infty(\mathcal{X})
$
obtained by transporting a $\Gamma$-periodic function
to the quotient by $\Gamma$.
It is also well known that the spaces $\mathscr{S}^\prime_\Gamma(\X)$ and
$\mathscr{S}^\prime(\mathbb{T})$ can be identified by a natural
topological isomorphism.
\end{remark}

\begin{hypothesis}\label{Hyp-h}
 From now on we shall consider a fixed real elliptic symbol $h\in
S^m_1(\X)_\Gamma$ for some $m>0$.
\end{hypothesis}
Let us recall that for
a $h\in S^m_1(\X)$ which is real, elliptic, and with $m>0$, the operator
$\mathfrak{Op}(h)$ has
a self-adjoint extension $H$ with domain
$\mathscr{H}^m(\X):=\left\{f\in L^2(\X)\mid(\bb1-\Delta)^{m/2}f\in
L^2(\X)\right\}$. Let us consider its resolvent $(H-\z)^{-1}$ at
$\z\in\mathbb{C}\setminus\sigma(H)$; it is known that
it is a pseudodifferential operator with a symbol $r_\z(h)$ of class
$S^{-m}_\rho(\X)$.
Let us notice that if $h\in S^m_{1}(\X)_{\Gamma}$, the operator
$\mathfrak{Op}(h)$ commutes with the action of $\Gamma$ through translations
and thus we can define its restriction to test functions on the torus:
\beq
\left.\mathfrak{
Op}(h)\right|_{\mathscr{S}_{\Gamma}(\X)}
:\mathscr{S}(\mathbb{T})\rightarrow\mathscr{S}(\mathbb{T}).
\eeq
This allows us to consider its self-adjoint extension $H_{\mathbb{T}}$ defined on
$\mathscr{H}^m(\mathbb{T})\cong\mathscr{S}
^\prime_{\Gamma}(\X)\cap\mathscr{H}_{\rm loc}^m(\X)$ and acting in the
Hilbert space
$L^2(\mathbb{T})\cong\mathscr{S}
^\prime_{\Gamma}(\X)\cap L_{\rm loc}^2(\X)$. 

Given a Hilbert space $\mathcal{K}$ and two vectors $u$ and $v$ in $\mathcal{K}$, we shall denote by 
$|u\rangle\langle v|$ the rank 1 operator in $\mathcal{K}$ associated to it, i.e. the operator
given by $|u\rangle\langle v|\; w:=u\langle v,w\rangle_\mathcal{K}$, for any $w\in\mathcal{K}$.

We shall need the following notation for different types of Fourier transforms:
$$
\mathcal{F}_{\Xi}\phi(\overline{X}):=(2\pi)^{-d}\int_{\Xi}e^{i\sigma(X,Y)}
\phi(Y)\,dY,\ \forall\phi\in\mathscr{S}(\Xi);\qquad\overline{(x,\xi)}:=(\xi,x);
$$
$$
\mathcal{F}_{\X}\phi(\xi):=(2\pi)^{-d/2}\int_{\X}e^{-i\langle \xi,x\rangle }\phi(x)\,dx,\ \forall\phi\in\mathscr{S}(\X);
$$
$$
\mathcal{F}_{\Gamma}\underline{v}(\theta):=|E|^{-d/2}\underset{\gamma\in\Gamma}{
\sum}\underline{v}(\gamma)e^{-i\langle \theta,\gamma\rangle },\qquad\forall\underline{v}\in
l^1(\Gamma).
$$
For any function $\mathbb{T}_*\ni\theta\mapsto\varphi(\theta)\in\mathbb{C}$ we denote by 
$\widetilde{\varphi}(\xi):=\varphi(\pi_*(\xi))$ its periodic extension to $\X^*$.  

Finally let us recall the \textbf{Poisson formula} (\cite{Ho1}) that we shall use several times: 
\begin{equation}\label{II.4.7}
\frac{(2\pi)^d}{|E_*|}\underset{\gamma\in\Gamma}{\sum}
\delta_\gamma\ =\
\underset{\gamma^*\in\Gamma_*}{\sum}\sigma_{\gamma^*},
\end{equation}
as distributions in
$\mathscr{S}^\prime(\X)$ with series converging in the weak sense.

\subsection{The Bloch-Floquet transform}
\label{B-bands}

Let us recall some basic facts concerning the
Bloch-Floquet method and fix some notation that will be used in the following.
We consider
\beq
\mathscr{F}\cong\left\{\hat{F}\in L^2_{\text{\sf
loc}}(\X\times\X^*)\,\mid\,\tau_\gamma\hat{F}=\sigma_{\gamma}\hat{F}\,
\ \forall\gamma\in\Gamma , \tau_{\gamma^*}\hat{F}=\hat{F}
\ \forall\gamma^*\in\Gamma^*\right\}
\eeq
with the Hilbertian norm
$\left\|\hat{F}\right\|^2:=\int_{E}\int_{
E_*}\left|\hat{F}(x,\xi)\right|^2\,d\xi\,dx$ and the {\it the Bloch-Floquet}
unitary map
\beq
\mathscr{U}_\Gamma:L^2(\X)\rightarrow\mathscr{F},\quad
\big(\mathscr{U}_\Gamma
f\big)(x,\xi)=\underset{\gamma\in\Gamma}{\sum}
\sigma_\gamma(\xi)f(x-\gamma)
\eeq
with its inverse having the explicit form
\beq
\left(\mathscr{U}_\Gamma^{-1}\hat{F}\right)(x_0+\gamma)=\left|\mathbb{T
}_*\right|^{-1}\int_{\mathbb{T}_*}\sigma_\gamma(\theta)\hat{F}(x_0,
\theta)d\theta.
\eeq
\begin{remark}\label{tau-U-sigma}
 For any $\gamma\in\Gamma$ we notice that $\mathscr{U}_\Gamma\tau_\gamma=\sigma_{\gamma}\mathscr{U}_\Gamma$.
\end{remark}

\vspace{0.2cm} 

Let us briefly recall the {\it direct integral} structure of $\mathscr{F}$ that
is very useful in the study of periodic $\Psi$DO. Using Fubini Theorem we
consider each function in $\mathscr{F}$ as a function defined on $\mathbb{T}_*$
and taking values in $L^2_{\text{\sf loc}}(\X)$ and define for each
$\theta\in\mathbb{T}_*$ the complex linear space
\beq
\mathscr{F}_\theta:=\left\{f\in L^2_{\text{\sf loc}}(\X)\,\mid\,\tau_\gamma
f=\sigma_{\gamma}(\theta)f\equiv\sigma_{\theta}(\gamma)f\right\}.
\eeq
For any $\theta\in\mathbb{T}_*$ we define the Hilbertian norm
$\|f\|_{\mathscr{F}_\theta}^2=\int_E|f(x)|^2dx$ and we can define a family of
unitary operators
$\mathscr{V}_\theta:\mathscr{F}_\theta\overset{\simeq}{\rightarrow}
L^2(\mathbb{T})$ by the formula
$\mathscr{V}_\theta\hat{F}\big(\pi(x)\big):=\sigma_{-\theta}(x)\hat{F}(x,
\theta)$. 
Following the procedure in \cite{Dix} one can consider on the field of Hilbert
spaces $\{\mathscr{F}_\theta\}_{\theta\in\mathbb{T}_*}$ over the compact smooth
manifold $\mathbb{T}_*$ the measurable structure defined by the constant field
$L^2(\mathbb{T})$ and the field of unitaries
$\{\mathscr{V}_\theta\}_{\theta\in\mathbb{T}_*}$ and define the associated direct
integral of
Hilbert spaces $\int_{\mathbb{T}_*}^\oplus\mathscr{F}_\theta\,d\theta$ that will
be unitarily isomorphic to $\mathscr{F}$.

For any $f\in\mathscr{S}(\X)$ we can write
$$
\big(\mathscr{U}_\Gamma\mathfrak{Op}(h)f\big)(\theta,x)\ =\
\underset{\gamma\in\Gamma}{\sum}e^{-i<\theta,\gamma>}\big(\mathfrak{Op}
(h)f\big)(x-\gamma)=\Big[\left.\mathfrak{Op}(h)\right|_{\mathscr{F}
_\theta}\Big(\underset{\gamma\in\Gamma}{\sum}e^{-i<\theta,\gamma>}\big(\tau_{-\gamma}
f\big)\Big)\Big](x),
$$
with the series converging for the weak topology of tempered distributions.
Thus
\beq\label{op-h-theta}
\mathscr{U}_\Gamma\mathfrak{Op}(h)
\mathscr{U}_\Gamma^{-1}\ =\
\int_{\mathbb{T}_*}^\oplus\left.\mathfrak{Op}(h)\right|_{\mathscr{F}
_\theta}\,d\theta.
\eeq

\begin{definition}\label{defgama}
We consider the following spaces of tempered distributions on $\Xi$:
\begin{enumerate}
 \item
$
\mathscr{S}^\prime_\Gamma(\Xi):=\left\{F\in\mathscr{S}
^\prime(\Xi)\mid
\tau_\gamma F=\sigma_{\gamma}F,\forall\gamma\in\Gamma,\
\tau_\gamma^*F=F,\forall\gamma^*\in\Gamma_*\right\},
$
\item
$
\mathscr{F}_{(s)}:=\left\{F\in\mathscr{S}^\prime_\Gamma(\Xi)\mid
\big((\bb1-\Delta)^{s/2}\otimes\bb1\big)F\in\mathscr{F}\right\},
\ \forall s\in\mathbb{R}
$
\item
$
\mathscr{F}_{(s),\theta}:=\left\{f\in\mathscr{F}_\theta\mid
(\bb1-\Delta)^{s/2}f\in\mathscr{F}_\theta\right\},
\ \forall s\in\mathbb{R},\ \forall\theta\in\mathbb{T}_*.
$
\end{enumerate}
\end{definition}

The operator $\hat{H}:=\mathscr{U}_\Gamma H\mathscr{U}_\Gamma^{-1}$
acting in the Hilbert space $\mathscr{F}$ is self adjoint on the domain
$\mathscr{F}_{(m)}$. We notice that the operator
$\mathscr{U}_\Gamma\mathfrak{Op}(h)\mathscr{U}_\Gamma^{-1}$ takes
any smooth function
$\mathbb{T}_*\ni\theta\mapsto f(\theta)\in\mathscr{F}_{(m),\theta}$ into
the smooth function
$\mathbb{T}_*\ni\theta\mapsto\left.\mathfrak{Op}
(h)\right|_{\mathscr{F}_\theta}f(\theta)\in\mathscr{F}_{\theta}$. We can also consider $\hat{H}(\theta)$ as the self-adjoint operator defined by 
$\left.\mathfrak{Op}(h)\right|_{\mathscr{F}
_\theta}$ in the Hilbert space $\mathscr{F}_\theta$ with domain
$\mathscr{F}_{(m),\theta}$ and due to its ellipticity we have that
$\mathscr{F}_{(m),\theta}=\left\{f\in\mathscr{F}_\theta\mid\left(\left.
\mathfrak{Op}(h)\right|_{\mathscr{F}_\theta}\right)f\in\mathscr{F}
_\theta\right\}$.
On these subspaces the quadratic graph-norms
$
\|f\|_{(m),\theta}\ :=\ \left\|\big(\bb1-\Delta\big)^{m/2}f\right\|_{L^2(E)}^2
$
define a field of Hilbert spaces that are unitarily transformed into
$\mathcal{H}^m(\mathbb{T})$ through the unitaries
$\{\mathscr{V}_\theta\}_{\theta\in\mathbb{T}_*}$. Using once again the theory in
\cite{Dix} we can define the direct integral Hilbert space
$\int_{\mathbb{T}_*}^\oplus\mathscr{F}_{(m),\theta}\,d\theta$ that will be
unitarily isomorphic to $\mathscr{F}_{(m)}$. Due to the above arguments, we can
also define
the measurable field of bounded operators 
$
\left\{\hat{H}(\theta)\in\mathbb{B}\big(\mathscr{F}_{(m),\theta};\mathscr{F}
_\theta\big)\right\}_{\theta\in\mathbb{T}_*}
$
and the associated direct integral
$$
\hat{H}=\mathscr{U}_\Gamma
H\mathscr{U}_\Gamma^{-1}=\int_{\mathbb{T}_*}^\oplus
\hat{H}(\theta)d\theta\ \in\
\mathbb{B}\big(\mathscr{F}_{(m)};\mathscr{F}\big).
$$

One has the identity
$
\sigma_{-\theta}\left.\mathfrak{Op}(h)\right|_{\mathscr{F}
_\theta}\sigma_{\theta}=\left.\mathfrak{Op}(\tau_\theta h)\right|_{\mathscr{F}_0}
$ as operators
acting in $L^2(\mathbb{T})$. For the convenience of the reader we remind the following fundamental result:
\begin{description}\item[Bloch-Floquet Theorem.] Under Hypothesis \ref{Hyp-h} the self-adjoint operator $H$ has the following spectral properties:
 \begin{enumerate}
  \item There exist a family of continuous functions
$
\mathbb{T}_*\ni\theta\mapsto\lambda_j(\theta)\in\mathbb{R}$  indexed by $j\in\mathbb{N}^*
$
such that $\lambda_j(\theta)\leq\lambda_{j+1}(\theta)$ for every $j\in\mathbb{N}^*$ and $\theta\in\mathbb{T}_*$, and
$$
\sigma\big(\hat{H}(\theta)\big)=\underset{j\in\mathbb{N}^*}{\bigcup}\{\lambda_j(\theta)\}.
$$
\item There exists a family of measurable functions
$
\mathbb{T}_*\ni\theta\mapsto\phi_j(\theta)\in\mathscr{F}_{(m),\theta}$ indexed by $j\in\mathbb{N}^*
$
such that $\|\phi_j(\theta)\|_{\mathscr{F}_\theta}=1$ and
$$
\hat{H}(\theta)\phi_j(\theta)=\lambda_j(\theta)\phi_j(\theta),\qquad\forall j\in\mathbb{N}^*, \ \forall\theta\in\mathbb{T}_*.
$$ 
 \end{enumerate}
\end{description}

\begin{hypothesis}\label{is-sp-band}
There exists a bounded open interval $I\subset\mathbb{R}$ such that
$I\,\cap\,\sigma(H)\,=:\,\sigma_I(H)\neq \emptyset$ and 
$\text{\sf dist}\big(I,\sigma(H)\setminus I\big)=d_0>0$.
We say that $\sigma_I(H)$ is an \textit{isolated spectral island}, completely included in $I$. We notice that in this case there exist
$(j,N)\in\mathbb{N}^*\times\mathbb{N}^*$ such that
$\sigma_I(H)=\underset{j+1\leq k\leq
j+N}{\bigcup}\lambda_k\big(\mathbb{T}_*\big)$. We shall denote
by $J_I:=\{j+1,\ldots,j+N\}$ the set of indices of the eigenvalues contained
in the isolated spectral island $\sigma_I(H)\subset I$.
\end{hypothesis}

\vspace{0.5cm}

\begin{definition}\label{Is-sp-band-op}
Let us choose a smooth, closed and positively oriented simple contour $\mathcal{C}$ in the complex plain, which surrounds $I$ and does not intersect the spectrum 
$\sigma(H)$. We introduce the band projection and the band Hamiltonian:
\begin{equation}\label{Prop-bandOp}
\mathfrak{p}\ :=E_I(H)=\ \frac{i}{2\pi}\int_\mathcal{C}(H-\z)^{-1}d\z,\quad \mathcal{H}_I:=\mathfrak{p}\mathcal{H},\quad \mathfrak{H}\ :=\ \frac{i}{2\pi}\int_\mathcal{C}\z(H-\z)^{-1}d\z.
\end{equation}
For $ j\in\mathbb{N}^*$ we define {\it the Bloch eigenprojections}: $\mathfrak{p}_j:=\mathscr{U}_\Gamma^{-1}\left(\int_{\mathbb{T}_*}^\oplus
\hat{\mathfrak{p}}(\theta)d\theta\right)\mathscr{U}_\Gamma$ where $\hat{\mathfrak{p}}(\theta):=|\phi_j(\theta)\rangle\langle\phi_j(\theta)|$ are 1-dimensional orthogonal projections in $\mathscr{F}_\theta$.
\end{definition}

\vspace{0.5cm}

We shall learn from Corollary \ref{reg-symbol} that both operators $\mathfrak{p}$
and $\mathfrak{H}$ have associated symbols $S_{\mathfrak{p}}$
and resp. $S_{\mathfrak{H}}$ of class $S^{-\infty}(\X)$. In general, this is not the case for each individual 
$\mathfrak{p}_j$ (even for $j\in J_I$). Also, the following statements hold true:
\begin{equation}\label{Prop-0}
[H,\mathfrak{p}_j]=0,\quad HE_I(H)=\underset{j\in J_I}{\sum}\mathfrak{Op}(\lambda_j)\mathfrak{p}_j.
 \end{equation}

\subsection{Our main results}
\label{main-res}

Assume that Hypothesis \ref{Hyp-h} is satisfied. 
Recall that a magnetic field is described by a closed 2-form
$B\equiv\sum\limits_{j,k=1}^NB_{jk}(x)dx_j\wedge dx_k$ with $B_{jk}$ defining an
antisymmetric matrix that verifies the equation ${\rm
d}B=0$ (here ${\it d}$ denotes the exterior derivative on differential forms). Let us
point out that we make no periodicity or slow variation assumption on the magnetic
field. 

\begin{hypothesis}\label{Hyp-B}
 We shall only consider $B_{jk}\in BC^\infty(\X)$ for any $j<k$.
 \end{hypothesis}
 
 \vspace{0.5cm}
 
We shall consider the limit of {\it weak} magnetic fields,
controlled by a small parameter $\epsilon\in[0,\epsilon_0]$ (see Hypothesis \ref{Hyp-V-magn} for more precise details). Roughly speaking  we shall consider a family of magnetic fields
$\{B_\epsilon\}_{\epsilon\in[0,\epsilon_0]}$ of the form
$
B_\epsilon:=\epsilon\,B^0_\epsilon
$, where the components of $B^0_\epsilon$ belong to $
BC^\infty\big(\X\big)$ uniformly with respect to $\epsilon\in[0,\epsilon_0]$. There exists an $1$-form
$A:=\sum\limits_{j=1}^{N}A_j(x)dx_j$ called {\it a vector potential} such that
$B={\rm d}A$. This vector potential is not unique and any
$A^\prime:=A+{\rm d}f$ with $f\in C^2(\X)$ also verifies  the equation $B={\rm
d}A^\prime$. It is well known that under our hypothesis on $B$ one can always
choose $A$ such that $A_j\in C^\infty_{\text{\sf pol}}(\X)$.  We shall fix a family of vector potentials 
$\{A^0_\epsilon\}_{\epsilon\in (0,\epsilon_0]}$ having components in a bounded subset of 
$C^\infty_{\text{\sf pol}}(\X)$ such that $B^0_\epsilon=dA^0_\epsilon$; 
then $B_\epsilon=dA_\epsilon$ for $A_\epsilon:=\epsilon A^0_\epsilon$.
    
The basic mathematical construction we use is the '{\it twisted
pseudodifferential calculus}' (see \eqref{OpA}) introduced in \cite{MP1,MPR,MPR1} (or \cite{Ne05} for its integral kernel version)
and developped
in \cite{IMP1,IMP2,ML}, that associates a '{\it quantized operator}' $\mathfrak{Op}^A(F)$ to any
H\"{o}rmander type symbol $F\in S^m_\rho(\X)$. We shall use the shorthand notation $\mathfrak{Op}^\epsilon(F):=\mathfrak{Op}^{A_\epsilon}(F)$.

Recalling the results in \cite{IMP1}, let $H^\epsilon$ be the self-adjoint
extension of $\mathfrak{Op}^\epsilon(h)$ with the domain given by the magnetic
Sobolev space. Using the results in
\cite{AMP, IP14} or \cite{B, CP-1, CP-2} we know that the resolvent set is stable for small variations of $\epsilon$. 
More precisely, given $I$ as in Hypothesis \ref{is-sp-band}, there exists  
$\epsilon_0>0$ small enough, such that 
$\sigma(H^\epsilon)\cap I\ne\emptyset$ and 
$\text{\sf dist}\big(I,\sigma(H^\epsilon)\setminus I\big)\geq d_0/2$ for any
$\epsilon\in[0,\epsilon_0]$. Moreover, the Hausdorff distance between $\sigma(H^\epsilon)\cap I$ and $\sigma_I(H)$ goes to zero with $\epsilon$.

The main question we are concerned with is the following:
{\it if $\epsilon\in[0,\epsilon_0]$ and $\epsilon_0>0$ is small enough, can we
replace $H$ by $H^\epsilon$ in \eqref{Prop-0} putting in the
right hand side a power series in $\epsilon$ with leading term similarly defined 
but with $\mathfrak{Op}^\epsilon$ instead of $\mathfrak{Op}$?} Our first result is contained in Theorem \ref{T-A} and gives a
partial answer to the above question. In order to state it, we shall need the
following technical result:

\begin{proposition}\label{rez-dev}
For every $\z\in \rho(H^\epsilon)$ let $r^\epsilon_\z(h)\in S^{-m}_1(\X)$ denote
the symbol of $(H^\epsilon-\z)^{-1}$ (as defined in \cite{IMP2}). If $\z$
belongs to a 
compact subset $K$ of
$\mathbb{C}\setminus\sigma(H)$, there exists
$\epsilon_0>0$ such that for $\epsilon\in[0,\epsilon_0]$ we
have that:
\begin{enumerate}
\item $K\subset\mathbb{C}\setminus\sigma(H^\epsilon)$;
\item the following development, convergent in the topology of the $C^*$-norm
$\|\cdot\|_{W,\epsilon}$ (induced from
$\mathbb{B}(\mathcal{H})$ as defined in \eqref{iulie4}) is true, uniformly with respect to
$(\epsilon,\z)\in[0,\epsilon_0]\times K$:
$$
r^\epsilon_\z(h)\ =\
\underset{n\in\mathbb{N}}{\sum}\epsilon^nr_n(h;\epsilon,\z),\quad
r_0(h;\epsilon,\z)=r^0_\z(h),\quad r_n(h;\epsilon,\z)\in S^{-(m+2n)}_1(\X);
$$
\item the map $K\ni \z\mapsto
r^\epsilon_\z(h)\in
S^{-m}_1(\X)$ is continuous for the Fr\'{e}chet topology on $S^{-m}_1(\X)$
uniformly in $\epsilon\in[0,\epsilon_0]$.
\end{enumerate}
\end{proposition} 

\vspace{0.5cm}

\begin{theorem}\label{T-A}
 Under Hypothesis \ref{Hyp-h}, \ref{is-sp-band} and \ref{Hyp-V-magn} there exists $\epsilon_0>0$ small enough such that for any $\epsilon\in[0,\epsilon_0]$ and for any $n\in\mathbb{N}^*$
 we have that:
 \begin{enumerate}
  \item $H^\epsilon E_I(H^\epsilon)=\mathfrak{Op}^\epsilon\big(S_{\mathfrak{H}}\big)+\underset{1\leq k\leq n-1}{\sum}\epsilon^k\mathfrak{Op}^\epsilon\big(v^\epsilon_k\big)
  +\epsilon^n\mathfrak{Op}^\epsilon\big(R^\epsilon_I(h;n)\big)$ with:
  \begin{enumerate}
   \item $S_{\mathfrak{H}}(x,\xi):=\frac{(2\pi)^{d}}{|E_*|}\int_\X e^{-i<\xi,y>}\left[\int_{\mathbb{T}_*}\left(\underset{j\in J_I}{\sum}
   \lambda_j(\theta)\phi_j(x+y/2,\theta)\overline{\phi_j(x-y/2,\theta)}\right)d\theta\right]dy$,
   \item $v^\epsilon_k:=-(2\pi i)^{-1}\int_{\mathcal{C}}\z r_k(h;\epsilon,\z)d\z$,
   \item $R^\epsilon_I(h;n):=-(2\pi i)^{-1}\int_{\mathcal{C}}\z\left(\underset{k\geq n}{\sum}\epsilon^kr_k(h;\epsilon,\z)\right)d\z$.
  \end{enumerate}
\item there exists an orthogonal projection $P^\epsilon_{I,n}$ in $L^2(\X)$ such that
\begin{enumerate}
 \item $\|E_I(H^\epsilon)-P^\epsilon_{I,n}\|\leq C_n\epsilon^n,\quad\|H^\epsilon[E_I(H^\epsilon)-P^\epsilon_{I,n}]\|\leq C_n(h)\epsilon^n$,
 \item $\|[H^\epsilon,P^\epsilon_{I,n}]\|\leq C_n(h)\epsilon^n$,
 \item $P^\epsilon_{I,n}=\mathfrak{Op}^\epsilon\big(S_\mathfrak{p}\big)+\underset{1\leq k\leq n-1}{\sum}
 \epsilon^k\mathfrak{Op}^\epsilon\big(u^\epsilon_k\big)
  +\epsilon^n\mathfrak{Op}^\epsilon\big(R^\epsilon_I(p;n)\big)$
 with
  \begin{enumerate}
   \item $S_\mathfrak{p}(x,\xi):=\frac{(2\pi)^{d}}{|E_*|}\int_\X e^{-i<\xi,y>}\left[\int_{\mathbb{T}_*}\left(\underset{j\in JI}{\sum}
   \phi_j(x+y/2,\theta)\overline{\phi_j(x-y/2,\theta)}\right)d\theta\right]dy$,
   \item $u^\epsilon_k:=-(2\pi i)^{-1}\int_{\mathcal{C}} r_k(h;\epsilon,\z)d\z$,
   \item $R^\epsilon_I(p;n)\in S^{-nm}_1(\X)_\Gamma$ uniformly for $\epsilon\in[0,\epsilon_0]$.
  \end{enumerate}
\end{enumerate}
 \end{enumerate}
\end{theorem}
\begin{corollary}
If the spectral island $\sigma_I$ is the image of an $N$-fold degenerate isolated Bloch band $\lambda:\mathbb{T}_*\rightarrow\mathbb{R}$ we have that
$$
H^\epsilon E_I(H^\epsilon)=\mathfrak{Op}^\epsilon(\lambda)E_I(H^\epsilon)+\underset{1\leq k\leq n-1}{\sum}\epsilon^k\mathfrak{Op}^\epsilon\big(v^\epsilon_k\big)
  +\epsilon^n\mathfrak{Op}^\epsilon\big(R^\epsilon_I(h;n)\big).
$$
\end{corollary}

\vspace{0.5cm}

The second result is about the case in which the projection associated to the
non-magnetic spectral island admits a basis of localized Wannier functions (see
formula \eqref{Wannier-def} for their definition). Then we can associate to the
band Hamiltonian a smooth $N\times N$ matrix-valued function
$\{\mu_{jk}(\theta)\}$ defined on the dual torus $\mathbb{T}_*$ (see Definition \ref{D-M}); let us denote
by $\widetilde{\mu}$ its periodic extension to $\X^*$ and by
$\{(\widehat{\mu_{jk}})_\gamma\}_{\gamma\in\Gamma}$ its Fourier coefficients defining a sequence in $l^2(\Gamma)$
having rapid decay. 

We shall use the notation
\beq\label{iulie21}
\widetilde{\Lambda}^A(x,y)\ :=\ e^{-i\int_{[x,y]}A},
\eeq
where the integral is taken along the oriented segment $[x,y]$. For every fixed pair $1\leq j,k\leq N$ we consider the operator 
$\mathfrak{Op}^\epsilon_\Gamma(\widetilde{\mu}_{jk})$ defined in $l^2(\Gamma)$ by the infinite matrix 
\beq\label{iulie20}
\left[\mathfrak{Op}^\epsilon_\Gamma(\widetilde{\mu}_{jk})\right]_{\alpha\beta}\
:=\ \widetilde{\Lambda}^\epsilon(\alpha,\beta)\widehat{\mu_{jk}}_{\alpha-\beta}.
\eeq
The $N\times N$ matrix with entries 
$\mathfrak{Op}^\epsilon_\Gamma(\widetilde{\mu}_{jk})$ for $1\leq j,k\leq N$ defines an operator in $l^2(\Gamma)^N$ that we denote by
$\mathfrak{Op}^\epsilon_\Gamma(\widetilde{\mu})$.

\begin{theorem}\label{teoremadoi}
 Under Hypothesis \ref{Hyp-h}, \ref{is-sp-band} and \ref{Hyp-V-magn},  if the isolated spectral band at zero magnetic field admits an orthonormal basis consisting of composite Wannier functions (see Hypothesis \ref{H.II.8.1}), and if $\epsilon_0$ is small enough, then there exists an orthonormal {\it magnetic  localized} basis 
 $\{\W^\epsilon_{\gamma,j}\}_{(\gamma,j)\in\Gamma\times J_I}$ such that:
 \begin{enumerate}
  \item $E_{I}(H^\epsilon)=\underset{\gamma\in\Gamma}{\sum}\underset{j\in J_I}{\sum}
  |\W^\epsilon_{\gamma,j}\rangle\langle\W^\epsilon_{\gamma,j}|$,
  \item $\underset{\gamma\in \Gamma}{\sup}\, \underset{x\in\X}{\sup}\langle x-\gamma\rangle ^m|\W^\epsilon_{\gamma,j}(x)|<\infty,\ \forall(\gamma,j)\in\Gamma\times J_I$,
  \item there exists a positive constant $C<\infty$ such that 
$$
  \left\|H^\epsilon
E_{I}(H^\epsilon)\,-\,\sum_{(\alpha,\beta)\in\Gamma\times\Gamma}\
\sum_{(j,k)\in J_I\times J_I}\
\left[\mathfrak{Op}^\epsilon_\Gamma(\widetilde{\mu}_{jk})\right]_{\alpha\beta}
\left|\W^\epsilon_{\alpha,j}\rangle\langle\W^\epsilon_{\beta,k}\right|\right\|_{
\mathbb { B } (L^2(\X))}\ \leq\ C\epsilon.
$$
 \end{enumerate}
\end{theorem}
 
\begin{corollary}\label{iulie33}
 Under our assumptions, if $\epsilon_0$ is small
enough, then the magnetic band Hamiltonian $H^\epsilon E_{I}(H^\epsilon)$ is
isospectral with $\mathfrak{Op}^\epsilon_\Gamma(\widetilde{\mu})$ up to an
error of order $\epsilon$, i.e. the Hausdorff distance between their spectra is
of order $\epsilon$. 
\end{corollary}

\vspace{0.5cm}

Our last result deals with the case of a slowly varying magnetic perturbation and makes the connection with the `usual' magnetic Weyl calculus via minimal coupling. 
\begin{definition}\label{Hyp-V-magn-slow}
We say that a magnetic field $B_\epsilon(x)$ is \textit{slowly varying} if it can be derived from a magnetic vector potential $A_\epsilon(x)=A(\epsilon x)$ where $A$ is a smooth vector potential having bounded derivatives of all strictly positive orders. In this case 
$$
B_\epsilon(x)\ =\ \epsilon\; dA(\epsilon x),\qquad\epsilon\in[0,\epsilon_0].
$$
\end{definition}

\begin{theorem}\label{teorema3}
Besides the assumptions of Theorem \ref{teoremadoi}, we consider a slowly varying magnetic field. Denote by $\mathfrak{Op}(\mu^{\epsilon})$ the usual Weyl quantization of the matrix-valued symbol $\mu^{\epsilon}_{jk}(x,\xi):=\mu_{jk}(\xi-A_\epsilon(x))$, acting on $L^2(\X)^N$. Then the Hausdorff distance between the spectrum of $H^\epsilon E_{I}(H^\epsilon)$ and $\mathfrak{Op}(\mu^{\epsilon})$ is of order $\epsilon$.
\end{theorem}

\section{The non-magnetic case}

\subsection{Kernels and symbols of $\Gamma$-periodic
operators}\label{Op-decomp}

Let us consider a general $\Gamma$-periodic bounded operator
$\mathfrak{T}\in\mathbb{B}\big(L^2(\X)\big)$ (i.e. a bounded operator that
commutes with all the translations $\{\tau_\gamma\}_{\gamma\in\Gamma}$). Then Remark \ref{tau-U-sigma} implies that its
Bloch-Floquet transform is a decomposable operator and we can write
\beq
\hat{\mathfrak{T}}:=\mathscr{U}_\Gamma\mathfrak{T}
\mathscr{U}_\Gamma^{-1}=\int^\oplus_{\mathbb{T}_*}\hat{\mathfrak{T}}(\theta)d\theta,\qquad
\hat{\mathfrak{T}}(\theta)\in\mathbb{B}\big(\mathscr{F}_\theta\big),\ \forall\theta\in\mathbb{T}_*.
\eeq

We recall that the elements of $\mathscr{F}_\theta$ are
completely determined by their restriction to $E\subset\X$ and this restriction
defines a unitary isomorphism $\mathscr{F}_\theta\cong L^2(E)$. Thus
each fibre operator $\hat{\mathfrak{T}}(\theta)$ has an associated distribution
kernel
$\hat{T}(\theta)\in\mathscr{D}^\prime(\X\times\X)$
that leaves $\mathscr{F}_\theta$ invariant. Thus it verifies the relations:
\beq\label{theta-condition-distr}
\big(\tau_\alpha\otimes\tau_\beta\big)\hat{T}(\theta)\ =\
e^{-i<\theta,\alpha-\beta>}\hat{T}(\theta),\quad\forall\theta\in\mathbb{T},
\,\forall(\alpha,\beta)\in\Gamma^2.
\eeq
We shall
always suppose that the map
$\mathbb{T}_*\ni\theta\mapsto
\hat{\mathfrak{T}}(\theta)\in\mathbb{B}\big(\mathscr{F}_\theta\big)$ is bounded
and measurable but much stronger hypothesis will be necessary in the sequel. We
shall mainly be interested in situations in which the integral kernel of 
$\hat{T}(\theta)$ is of class $BC^\infty\big(\X\times\X\big)$ and
\eqref{theta-condition-distr} implies that they satisfy the relation:
\beq\label{theta-condition}
\hat{T}(\theta,x+\alpha,y+\beta)\ =\
e^{-i<\theta,\alpha-\beta>}\hat{T}(\theta,x,y),\quad\forall(\theta,x,y)\in\mathbb{T}
\times\X^2 ,
\,\forall(\alpha,\beta)\in\Gamma^2,
\eeq
\begin{hypothesis}\label{H-F-dec-op}
 Suppose given a bounded $\Gamma$-periodic operator
$\mathfrak{T}$ such that in the Bloch-Floquet representation it defines a family
of integral kernels $\hat{T}(\theta)$ of class
$BC^\infty\big(\X\times\X\big)$ and such that the map
$\mathbb{T}_*\ni\theta\mapsto
\hat{T}(\theta)\in BC^\infty\big(\X\times\X\big)$ is bounded and measurable.
\end{hypothesis}

\vspace{0.5cm}

Let us compute the integral kernel of $\mathfrak{T}=\mathscr{U}_\Gamma^{-1}\hat{\mathfrak{T}}\mathscr
{U}_\Gamma$. Recall that every $x\in \X$ can be uniquely written as $[x]+\hat{x}$ with $[x]\in \Gamma$ and $\hat{x}\in E$. 
For any $f\in\mathscr{S}(\X)$ and $x\in\X$ we have
\begin{align}
\big(\mathfrak{T}f\big)(x)&=|E_*|^{-1}\int_{\mathbb{T}_*}d\theta
\big(\hat{\mathfrak{T}}
(\theta)\mathscr{U}_\Gamma
f\big)(\theta,x)=|E_*|^{-1}\int_{\mathbb{T}_*}d\theta\int_E
d\hat{y}\,\hat{T}(\theta;x,\hat{y})\big(\mathscr{U}_\Gamma
f\big)(\theta,\hat{y})\nonumber \\
&=|E_*|^{-1}\underset{\gamma\in\Gamma}{\sum}\int_{\mathbb{T}_*}d\theta\int_E
d\hat{y}\,\hat{T}(\theta;\hat{x},\hat{y})e^{-i<\theta,[x]-\gamma>}
f(\hat{y}+\gamma)\nonumber \\
&=
\underset{\gamma\in\Gamma}{\sum}\int_E
d\hat{y}\,\check{T}([x]-\gamma;\hat{x},\hat{y})f(\hat{y}+\gamma),\nonumber \\
\check{T}
(\alpha;\hat{x},\hat{y})&:=|E_*|^{-1}\int_{E_*}d\theta\,\hat{T}(\theta;\hat{x},\hat{y
})e^{-i<\theta,\alpha>}.\label{kernel-1}
\end{align}
Thus $\mathfrak{T}$ has the integral kernel:
\beq\label{T-int-kernel}
K_{\mathfrak{T}}(x,y):=\check{T}([x]-[y],\hat{x},\hat{y}):=|E_*|^{-1}\int_{E_*}d\theta\,\hat{T}(\theta;\hat{x},\hat{y
})e^{-i<\theta,[x]-[y]>}.
\eeq
Notice that this kernel may have some decay in the variable $x-y$ depending on
the regularity of
the map $\mathbb{T}_*\ni\theta\mapsto
\hat{T}(\theta)\in BC^\infty\big(\X\times\X\big)$, which we shall assume to be $C^\infty$. Indeed, by partial integration in \eqref{kernel-1} we obtain that $<x-y>^N K_{\mathfrak{T}}(x,y)$ is globally bounded 
for every $N\geq 1$.  


Let us write down the symbol of the operator $\mathfrak{T}$
\begin{align}\label{symb-FKernel}
S_{\mathfrak{T}}(z,\zeta)&=(2\pi)^{d}\int_\X
e^{-i\langle \zeta,v\rangle}K_{\mathfrak{T}}(z+v/2,z-v/2)dv\nonumber \\
&=(2\pi)^{d}\int_\X
e^{-i\langle \zeta,v\rangle}\left(|E_*|^{-1}\int_{\mathbb{T}_*}d\theta
\,\hat{T}\big(\theta;z+(v/2),z-(v/2)\big)\right)dv.
\end{align}

\begin{proposition}\label{reg-symbol}
 If $\mathfrak{T}$ is a bounded $\Gamma$-invariant operator satisfying Hypothesis \ref{H-F-dec-op}, 
 then its symbol $S_{\mathfrak{T}}$ belongs to $S^{-\infty}(\X)$.
\end{proposition}
\begin{proof}
The result follows from \eqref{symb-FKernel} and from the decay properties of $K_{\mathfrak{T}}(z+v/2,z-v/2)$ seen as a function of $v$.
\end{proof}

\begin{proposition}\label{rho}
 Any function $\rho\in C(\mathbb{T}_*)$ defines an operator
$\mathbf{M}_\rho\in\mathbb{B}(\mathscr{F})$ given by multiplication with
$\rho(\theta)$ in each
fiber space $\mathscr{F}_\theta$. Then, denoting by
$\tilde{\rho}\in C_\Gamma(\X^*)$ the periodic extension of $\rho$ to $\X^*$ we
have the equality:
$$
\mathscr{U}_\Gamma^{-1}\mathbf{M}_\rho\mathscr{U}_\Gamma\ =\
\mathfrak{Op}(\tilde{\rho}).
$$
\end{proposition}
\begin{proof}
For any $f\in\mathscr{S}(\X)$ we have
\begin{align*}
&\left(\mathscr{U}_\Gamma^{-1}\left(\int_{\mathbb{T}_*}
^\oplus\rho(\theta)d\theta\right)\mathscr{U}_\Gamma
f\right)(x)=|E_*|^{-1}\int_{
\mathbb{T}_*}\left(\rho(\theta)\underset{\gamma\in\Gamma}{\sum}e^{-i<\theta,
\gamma>}f(x-\gamma)\right)d\theta\\
&=|E_*|^{-1}\underset{\gamma\in\Gamma}{\sum}\int_{
\mathbb{T}_*}\left(\rho(\theta)e^{-i<\theta,
\gamma>}f(x-\gamma)\right)d\theta=|E_*|^{-1}\underset{\gamma\in\Gamma}{\sum}
\left(\int_{
\mathbb{T}_*}\rho(\theta)e^{-i<\theta,
\gamma>}d\theta\right)f(x-\gamma)\\
&=(2\pi)^{-d}\int_\X\int_{
\X^*}e^{i<\xi,z>}\tilde{\rho}(\xi)f(x+z)d\theta
dz=\big(\mathfrak{Op}(\tilde{\rho})f\big)(x).
\end{align*}
\end{proof}

\subsection{The symbols of the `band operators'.}
\label{SS-spband}
Recall that the non-magnetic operator $H$ has an isolated spectral island $\sigma_I$, see Hypothesis \ref{is-sp-band}. Let us list some  properties of the band operators appearing in Definition \ref{Is-sp-band-op}.  

\begin{proposition}\label{comm-band-op}~
The following facts hold true:

\noindent 1. We have that
$
\big[H,\mathfrak{p}\big]\ =\ \big[\mathfrak{H},\mathfrak{p}\big]\ =\ 0
$.

\noindent 2. Using Proposition \ref{reg-symbol} we may conclude that both operators $\mathfrak{p}$
and $\mathfrak{H}$ have associated symbols $S_{\mathfrak{p}}$ and
resp. $S_{\mathfrak{H}}$ of class $S^{-\infty}(\X)$. 

\noindent 3. Using the Bloch-Floquet transform (see subsection \ref{B-bands}) we can write the band operators as
direct integrals of some fibre operators given in terms of the Bloch eigenvalues
and eigenvectors:
\beq\label{BF-band-Ham}
\mathfrak{H}\ =\
\mathscr{U}_\Gamma^{-1}\left(\int_{\mathbb{T}_*}^\oplus\underset
{j\in J_I}{\sum}\lambda_j(\theta)\left|\phi_j(\theta)\rangle\langle
\phi_j(\theta)\right|\,d\theta\right)\mathscr{U}_\Gamma,
\eeq
\beq\label{BF-band-proj}
\mathfrak{p}\ =\
\mathscr{U}_\Gamma^{-1}\left(\int_{\mathbb{T}_*}^\oplus\underset
{j\in J_I}{\sum}\left|\phi_j(\theta)\rangle\langle
\phi_j(\theta)\right|\,d\theta\right)\mathscr{U}_\Gamma\ \equiv\ 
\mathscr{U}_\Gamma^{-1}\left(\int_{\mathbb{T}_*}^\oplus\hat{\mathfrak{p}}(\theta
)d\theta\right)\mathscr{U}_\Gamma,
\eeq
and we have the following formulae for their integral kernels and symbols:
\beq\label{KHI}
K_\mathfrak{H}(x,y)=\underset{k\in J_I}{\sum}
\left(|E_*|^{-1}\int_{\mathbb{T}_*}
\lambda_k(\theta)\phi_k(x,\theta)
\overline{\phi_k(y,\theta)}\,d\theta\right),
\eeq
\beq
K_\mathfrak{p}(x,y)=\underset{k\in
J_I}{\sum}\left(|E_*|^{-1}\int_{\mathbb{T}_*}\phi_k(x,\theta)
\overline{\phi_k(y,\theta)}\,d\theta\right),
\eeq
\beq\label{SH-0}
S_\mathfrak{H}(z,\zeta)=\underset{k\in
J_I}{\sum}\left((2\pi)^d|E_*|^{-1}\int_{\X}e^{-i<\zeta,v>}\int_{\mathbb{T}_*}
\lambda_k(\theta)\phi_k(z+v/2,\theta)
\overline{\phi_k(z-v/2,\theta)}\,d\theta\right)dv,
\eeq
\beq
S_\mathfrak{p}(z,\zeta)=\underset{k\in
J_I}{\sum}\left((2\pi)^d|E_*|^{-1}\int_{\X}e^{-i<\zeta,v>}\int_{\mathbb{T}_*}
\phi_k(z+v/2,\theta)
\overline{\phi_k(z-v/2,\theta)}\,d\theta\right)dv.
\eeq
\end{proposition}

\vspace{0.5cm}

We can define for any $j\in J_I$ the Bloch orthogonal projections
$\hat{\mathfrak{p}}_j(\theta):=\left|\phi_j(\theta)\rangle\langle
\phi_j(\theta)\right|$ and
$\mathfrak{p}_j:=\mathscr{U}_\Gamma^{-1}\left(\int_{\mathbb{T}_*}
^\oplus\hat {\mathfrak{p}}_j(\theta)d\theta\right)\mathscr{U}_\Gamma$.
Let us notice that in general these one dimensional projections have much less
regularity in the variable $\theta\in\mathbb{T}_*$ than the finite sums
above, and thus these objects are
not very useful without some stronger hypothesis.

In the case of an isolated spectral band we have that 
 $\mathfrak{p}\ =\ \underset{j\in
J_I}{\sum}\mathfrak{p}_j$ and also
\begin{equation}\label{II.6.8}
HE_I(H)\,=\,\underset{j\in
J_I}{\sum}\mathfrak{Op}(\lambda_j)\mathfrak{p}_jE_I(H),
\end{equation}
\begin{equation}\label{II.6.9}
e^{-itH}E_I(H)\,=\,\underset{j\in
J_I}{\sum}\mathfrak{Op}\big(e^{-it\lambda_j}\big)\mathfrak{p}_jE_I(H).
\end{equation}

\subsection{The composite Wannier basis.}
\label{band-smooth-basis}

\begin{hypothesis}\label{H.II.8.1}
Under our Hypothesis \ref{is-sp-band} of an isolated spectral band, we suppose
also that the subspace
$\mathcal{F}^I_\theta:=\hat{\mathfrak{p}}(\theta)\mathscr{F}_\theta$ (see Definition \ref{Is-sp-band-op}) has an orthonormal 
basis 
$\psi_j(\theta)\in\mathscr{F}_\theta\cap BC^\infty(\X)$ indexed by $j\in
J_I$ and such that the map
$\mathbb{T}_*\ni\theta\mapsto\psi_j(\theta)\in BC^\infty(\X)$ is smooth for any
$j\in J_I$.
We do not suppose that the vectors $\{\psi_j(\theta)\}$ of the
basis are eigenvectors of $\hat{H}(\theta)$; but we notice that they are finite linear combinations of such eigenvectors and thus smooth functions of $x\in\X$.
\end{hypothesis}

\vspace{0.5cm}

We have that $\psi_j(\theta)\in\hat{\mathfrak{p}}(\theta)\mathscr{F}_\theta\subset \mathcal{D}\big(\hat{H}(\theta)\big)$. This smooth basis allows us to
construct a so-called {\it composite Wannier basis} for the range of $\mathfrak{p}=E_I(H)$. 
\begin{proposition}\label{iulie30}
Define the functions
\beq\label{Wannier-def}
w_j:=\mathscr{U}_\Gamma^{-1}\psi_j\in\mathscr{S}(\X),\ \forall
j\in J_I;\qquad
w_j(x)=|E_*|^{-1}\int_{\mathbb{T}_*} \psi_j(\hat{x},\theta)e^{-i\langle \theta,[x]\rangle}d\theta,
\eeq
that are orthonormal in $L^2(\X)$.
Then the translated functions
$\mathcal{W}_{\gamma,j}:=\tau_{-\gamma}w_j$ with $(\gamma,j)\in\Gamma\times
J_I$ form an
orthonormal basis for $\mathfrak{p}\mathcal{H}$. Moreover, $<x>^m w_j(x)$ lies in $L^\infty(\X)$ for all $m\geq 0$.
\end{proposition}
\begin{proof}
We note that
$\big(\mathscr{U}_\Gamma\mathcal{W}_{\gamma,j}\big)(\hat{x},\theta)=e^{i\langle \theta,
\gamma\rangle}\psi_j(\hat{x},\theta)$ and the proposition follows from the fact that
$\{\psi_j(\theta)\}_{j\in J_I}$ is an orthonormal system in
$\mathscr{F}_\theta$ for each $\theta\in\mathbb{T}_*$. 

We assumed that each $\psi_j(\hat{x},\theta)$ is smooth in $\theta$ which implies that $<x>^m w_j(x)$ lies in $L^2(\X)$ for all $m\geq 0$.  Moreover, because the integral kernel of $\mathfrak{p}$ obeys the estimate $|K_\mathfrak{p}(x,y)|\leq C_m <x-y>^{-m}$ valid for every $m\geq 0$, it follows that $<x>^m w_j(x)$ lies in $L^\infty(\X)$ for all $m\geq 0$.
\end{proof}

\vspace{0.5cm}

\begin{definition}\label{D-M}
\begin{enumerate}
\item
We define {\it the effective Hamiltonian} associated to $\sigma_I$ to be the $N\times N$ matrix valued map 
$$
\mathbb{T}_*\ni\theta\mapsto\mu_{jk}(\theta):=\langle\psi_j(\theta),
\hat{H}(\theta)\psi_k(\theta)\rangle_{\mathscr{F}_\theta}\in\mathbb{C},\quad (j,k)\in J_I\times J_I.
$$

\item
For any $(j,k)\in J_I\times J_I$ we define the rank-one
operators:
$
\pi_{jk}(\theta):=\left|\psi_k(\theta)\rangle\langle\psi_j(\theta)\right|
$
and their associated $\Gamma$-invariant operators in $L^2(\X)$ given by 
$
\Pi_{jk}:=\hat{\mathscr{U}}_\Gamma^{-1}\pi_{jk}\hat{\mathscr{U}}_\Gamma$.
We shall denote by $P_{jk}$ the integral kernel of the operator $\Pi_{jk}$.
\end{enumerate}
\end{definition}

\vspace{0.5cm}

\begin{remark}\label{R-Pjk}
 We have the following properties:
 \begin{enumerate}
  \item $\{\pi_{jk}\}_{(j,k)\in J_I^2}$ is a
family of bounded  operators on $\mathscr{F}$ that satisfy the relations:
$$
\pi_{jk}^*=\pi_{kj}.\quad\forall(j,k)\in J_I^2;\qquad
\pi_{jk}\pi_{pq}=\delta_{jq}\pi_{pk}.\quad\forall(j,k,p,q)\in
J_I^4.
$$
\item Similar properties are also valid for the family
$\{\Pi_{jk}\}_{(j,k)\in J_I^2}$. Moreover, for every pair $(j,k)\in J_I^2$, there exists a symbol $p_{jk}\in S^{-\infty}(\X)$ such that $\Pi_{jk}=\mathfrak{Op}(p_{jk})$
and
we
have the explicit formula
$$
p_{jk}(z,\zeta)=(2\pi)^d|E_*|^{-1}\int_{\X}e^{-i<\zeta,v>}\int_{\mathbb{T}_*}
\psi_k\big(z+(v/2),\theta\big)\overline{
\psi_j\big(z-(v/2),\theta\big)}d\theta\,dv,\quad\forall(z, \zeta)\in\Xi.
$$
\item $\hat{\mathfrak{p}}=\underset{j\in J_I}{\sum}\pi_{jj}$ and
$\hat{\mathfrak{p}}\hat{H}\hat{\mathfrak{p}}\,=\,\int_{\mathbb{T}_*}
^\oplus\underset{(j,k)\in
J_I^2}{\sum}\mu_{jk}(\theta)\pi_{kj}(\theta)\,d\theta$.
 \end{enumerate}
\end{remark}

We also list without proof a few properties of the band Hamiltonian seen as a matrix in the Wannier basis:

\begin{proposition}\label{Band-operators}
 The following equalities are true (here 
$\widetilde{\mu}_{jk}$ is the periodic extension of $\mu_{jk}$):
\beq\label{PO-matrix}
\mathfrak{p}\,=\,N^{-1}\underset{(j,k)\in
J_I^2}{\sum}\mathfrak{Op}(p_{jk})\mathfrak{Op}(p_{kj});\qquad\mathfrak{H}=\mathfrak{
p } H\mathfrak {
p}\,=\,N^{-1}\underset{(j,k)\in
J_I^2 }{\sum}\mathfrak{Op}(\widetilde{\mu}_{jk})\mathfrak{Op}(p_{kj});
\eeq
$$
\left\langle\mathcal{W}_{\alpha,j}\,,\,\Pi_{lm}\mathcal{W}_{\beta,k}
\right\rangle_
{ L^2(\X) }\ =\
\delta_{\alpha\beta}\delta_{jm}\delta_{kl};
$$
$$
\left\langle\mathcal{W}_{\alpha,j}\,,\,H\mathcal{W}_{\beta,k}\right\rangle_{
L^2(\X)}=|E_*|^{-1}\int_{\mathbb{T}_*}\mu_{jk}(\theta)e^{i\langle\theta, \alpha-\beta\rangle }d\theta=:\widehat{\mu_{jk}}_{\alpha-\beta};
$$
$$
\left\langle\mathcal{W}_{\alpha,j}\,,\,e^{-itH}\mathcal{W}_{\beta,k}
\right\rangle_{ L^2(\X)}=|E_*|^{-1}\int_{\mathbb{T}_*}[e^{-it\mu(\theta)}]_{jk}e^{i\langle\theta, \alpha-\beta\rangle }d\theta.
$$
\end{proposition}

\section{Adding a weak magnetic field}
\subsection{Brief recall of the magnetic Weyl calculus.}
\label{M-PsiDO}
We shall very briefly recall some of the main definitions and results
concerning a completely gauge covariant version of the `minimal coupling'
procedure developed in \cite{MP1,IMP1,IMP2,Ne05}.

Given a magnetic field $B$ and a choice of a vector potential $A$ for it and considering the
fundamental set of dynamical observables given by the '{\it minimal coupling}' hypothesis:
\beq\label{position}
\{Q_1,\ldots,Q_d\},\qquad\big(Q_jf\big)(x):=x_jf(x),\ \forall f\in\mathscr{S}(\X)
\eeq
\beq\label{magn-momenta}
\{\Pi^A_1,\ldots,\Pi^A_d\},\qquad\big(\Pi^A_jf\big)(x):=\big(-i\partial_jf\big)(x)-A_j(x)f(x),\ \forall f\in\mathscr{S}(\X),
\eeq
in \cite{MP1} we considered the {\it twisted Weyl system} defined by the unitary groups 
associated to the self-adjoint extensions of the above operators.
This procedure allows us to define a `{\it twisted pseudodifferential calculus}' (introduced in \cite{MP1,MPR,MPR1} 
and developped in \cite{IMP1,IMP2,ML}) that associates to any 
H\"{o}rmander type symbol $F\in S^m_\rho(\X)$ the following operator in $L^2(\X)$ (for all $u\in\mathscr{S}(\X)$ and $x\in\X$):
\beq\label{OpA}
\big(\mathfrak{Op}^A(F)u\big)(x)\ :=\
(2\pi)^{-d}\int_\X\int_{\X^*}e^{i\langle \xi,x-y\rangle}e^{-i\int_{[x,y]}A}F\big(\frac{x+y}
{2},\xi\big)u(y)\,d\xi\,dy.
\eeq
 

Two important results in \cite{MP1} state that two vector potentials that are
gauge equivalent define two unitarily equivalent functional calculi and that
the application $\mathfrak{Op}^A$ defined above extends to a linear and
topological isomorphism between $\mathscr{S}^\prime(\Xi)$ and
$\mathbb{B}\big(\mathscr{S}(\X);\mathscr{S}^\prime(\X)\big)$. At this point we
can make the connection with the `{\it twisted integral kernels}' formalism in
\cite{Ne05}, where for any integral kernel $K\in\mathscr{S}^\prime(\X\times\X)$
one associates a twisted integral kernel (see \eqref{Prop-bandOp})
\beq\label{twist-int-kernel}
K^A(x,y)\ :=\ \widetilde{\Lambda}^A(x,y)K(x,y).
\eeq
For any integral
kernel $K\in\mathscr{S}^\prime(\X\times\X)$ let us denote by
$\mathcal{I}\text{\sf nt}K$ its corresponding linear operator on
$\mathscr{S}(\X)$ (i.e. $\big(\mathcal{I}\text{\sf
nt}Ku\big)(x)=\int_\X K(x,y)u(y)\,dy$). Let us recall the usual Weyl calculus,
that we shall denote by $\mathfrak{Op}$ and the linear bijection
$\mathfrak{W}:\mathscr{S}^\prime(\Xi)\rightarrow\mathscr{S}^\prime(\X\times\X)$
associated to it by $\mathfrak{Op}(F)=\mathcal{I}\text{\sf nt}(\mathfrak{W}F)$: 
\beq\label{iulie3}
\big(\mathfrak{W}F\big)(x,y)\ :=\
(2\pi)^{-d}\int_{\X^*}e^{i<\xi,x-y>}F\big(\frac{x+y}{2},\xi\big)d\xi.
\eeq
Then we have the equality
\beq\label{magn-quant}
\mathfrak{Op}^A(F)\ =\ \mathcal{I}\text{\sf
nt}(\widetilde{\Lambda}^A\mathfrak{W}F).
\eeq

\begin{remark}\label{Omega}
We notice that under our Hypothesis \ref{Hyp-B} and with the notation \eqref{Prop-bandOp}:
$$
\widetilde{\Lambda}^A(x,z)\widetilde{\Lambda}^A(z,y)\widetilde{\Lambda}
^A(y,x)\,=\,\exp\left\{-i\int_{<x,y,z>}\hspace*{-10pt}
B\right\}\,=:\,
\Omega^B(x,y,z);
$$
the above integral is taken on the positively oriented triangle $<x,y,z>$. We have the estimation:
$$
\left|\Omega^B(x,y,z)-1\right|\leq C\|B\|_\infty
|(y-x)\wedge(z-x)|.
$$
\end{remark}

\vspace{0.5cm}

Using Theorem 4.1 in \cite{IMP1}, under Hypothesis \ref{Hyp-B}, for a symbol
$h\in S^m_1(\X)_\Gamma$ verifying Hypothesis \ref{Hyp-h}, the operator
$\mathfrak{Op}^A(h)$ for any $A$ with components of class $C^\infty_{\text{\sf
pol}}(\X)$ has a closure $H^A$ in $L^2(\X)$ that is self-adjoint on a domain
$\mathcal{H}^m_A$ (a `magnetic Sobolev space') and it is lower semibounded. Thus
we can define its resolvent $(H^A-\z)^{-1}$ for any $\z\notin\sigma(H^A)$ and
Theorem 6.5 in \cite{IMP2} states that it exists a well-defined symbol
$r^B_\z(h)\in S^{-m}_1(\X)$ such that 
$$
(H^A-\z)^{-1}\ =\ \mathfrak{Op}^A(r^B_\z(h)).
$$

Let us also recall from \cite{MP1} that the operator composition of the
operators $\mathfrak{Op}^A(F)$ and $\mathfrak{Op}^A(G)$ induces a {\it twisted
Moyal product} depending only on the magnetic field $B$:
\begin{align}\label{II.1.5}
\big(F\sharp^BG\big)(X)&:=\pi^{-2d}\int_\Xi dY\int_\Xi
dZ\,e^{-2i\sigma^\circ(Y,Z)}e^{-i\int_{T(x,y,z)}B}F(X-Y)G(X-Z)\\
&=\pi^{-2d}\int_\Xi dY\int_\Xi
dZ\,e^{-2i\sigma^\circ(X-Y,X-Z)}e^{-i\int_{\widetilde{T}(x,y,z)}B}F(Y)G(Z)\nonumber 
\end{align}
where we use the notation of the type $X=(x,\xi), Y=(y,\eta)$ etc.
and we have denoted by $T(x,y,z)$ the triangle in $\X$ of vertices $x-y-z,$,
$x+y-z$, $x-y+z$ and by $\widetilde{T}(x,y,z)$ the triangle in $\X$ of vertices
$x-y+z$, $y-z+x,z-x+y$. For any symbol $F$ we denote by $F^-_B$ its inverse
with respect to the magnetic Moyal product, if it exists. 

We shall often use the symbol $<\xi>^m$ for $m>0$ and we shall also need to
consider its magnetic Moyal inverse. For that we use the arguments in Section
2.1 of \cite{MPR1} and conclude that for $a>0$ large enough the symbol
$\mathfrak{s}_m(x,\xi):=<\xi>^m+a$, with $m>0$ has an inverse for the magnetic
Moyal product and we shall use the shorthand notation 
$\mathfrak{s}^B_{-m}$ instead of $\big(\mathfrak{s}_m\big)^-_B$ for this symbol with some fixed large enough $a_m>0$. 

For the completeness of our arguments we give the proof of a simplified version of Proposition
8.1 in \cite{IMP2} that will be important in our arguments.
\begin{proposition}\label{comp-symb}
 Suppose we are given $\phi\in S^m_\rho(\Xi)$, $\psi\in S^p_\rho(\Xi)$ and
$\theta\in BC^\infty(\mathcal{X};C^\infty_{\text{\sf pol}}(\mathcal{X}^2))$ (the
bounded smooth functions on $\mathcal{X}$ with values in the space of smooth
functions on $\mathcal{X}^2$ with polynomial growth together with their
derivatives). Then
$$
\mathfrak{L}(\theta;\phi,\psi)(X):=\int_\Xi dY\int_\Xi dZ
e^{-2i\sigma^\circ(Y,Z)}\theta(x,y,z)\phi(X-Y)\psi(X-Z)
$$
defines a symbol of class $S^{m+p}_\rho(\Xi)$ and the
mapping
$$
S^m_\rho(\Xi)\times
S^p_\rho(\Xi)\ni(\phi,\psi)\mapsto\mathfrak{L}(\theta;\phi,\psi)\in
S^{m+p}_\rho(\Xi)
$$
is continuous.
\end{proposition}
\begin{proof}
By a simple change of variables we can write
$$
\mathfrak{L}(\theta;\phi,\psi)(X)=\int_\Xi dY\int_\Xi dZ
e^{-2i\sigma^\circ(X-Y,X-Z)}\theta(x,x-y,x-z)\phi(Y)\psi(Z).
$$
For any natural numbers $N_1$, $N_2$, $M_1$, $M_2$ we have the identity:
\begin{align}\label{int-osc}
e^{-2i\sigma^\circ(X-Y,X-Z)}&=
\left(\frac{1-i\langle (\xi-\zeta),\partial_y\rangle }{1+2|\xi-\zeta|^2}\right)^{N_2}
\left(\frac{1+i\langle (\xi-\eta),\partial_z\rangle }{1+2|\xi-\eta|^2}\right)^{N_1}
\\
&\times
\left(\frac{1+i\langle (x-z),\partial_\eta\rangle }{1+2|x-z|^2}\right)^{M_2}\left(\frac{
1-i\langle (x-y),\partial_\zeta\rangle }
{1+2|x-y|^2}\right)^{M_1}\,e^{-2i\sigma^\circ(X-Y,X-Z)}\nonumber.
\end{align}
If
$\phi$ and $\psi$ are test functions, after integration by parts we obtain the estimate:
\begin{align*}
\left|\mathfrak{L}(\theta;\phi,\psi)(X)\right|&\leq C\left(\int_\Xi dY 
<\xi-\eta>^{-N_1}<\eta>^m<x-y>^{r_1(N_1,N_2)-M_1}\right)\\
&\times\left(\int_\Xi
dZ<\xi-\zeta>^{-N_2}<\zeta>^p<x-z>^{r_2(N_1,N_2)-M_2}\right)\leq
C^\prime<\xi>^{m+p},
\end{align*}
where we choose $N_1>|m|+n$, $N_2>|p|+n$, $M_1>r_1(N_1,N_2)+n$ and
$M_2>r_2(N_1,N_2)+n$, with $r_j(N_1,N_2)$ the powers dominating
$\partial_z^{N_1}\partial_y^{N_2}\theta(x,y-x,z-x)$.
Now let us compute the $\xi$-derivative of
$\mathfrak{L}(\theta;\phi,\psi)$:
\begin{align*}
\left(\partial_{\xi_j}\mathfrak{L}(\theta;\phi,\psi)\right)(X)
&=-\int_\Xi dY\int_\Xi dZ
\left[(\partial_{\eta_j}+\partial_{\zeta_j})e^{-2i\sigma^\circ(X-Y,X-Z)}\right]
\theta(x,x-y,x-z)\phi(Y)\psi(Z)\\
&=\int_\Xi dY\int_\Xi dZ
e^{-2i\sigma^\circ(X-Y,X-Z)}\theta(x,x-y,x-z)\left[\left(\partial_{\eta_j}
\phi\right)(Y)\right]\psi(Z)\\
&+\int_\Xi dY\int_\Xi dZ
e^{-2i\sigma^\circ(X-Y,X-Z)}\theta(x,x-y,x-z)\phi(Y)\left[\left(\partial_{
\zeta_j} \psi\right)(Z)\right]\\
&=\mathfrak{L}(\theta;(\partial_{\xi_j}\phi),\psi)(X)+\mathfrak{L}
(\theta;\phi,(\partial_{\xi_j}\psi))(X).
\end{align*}

Considering the $x$-derivative we obtain in a similar way that
$$
\left(\partial_{x_j}\mathfrak{L}(\theta;\phi,\psi)\right)(X)
=\mathfrak{L}(\theta;(\partial_{x_j}\phi),\psi)+\mathfrak{L}
(\theta;\phi,(\partial_{x_j}\psi))+\mathfrak{L}(\tilde{\theta}
;\phi,\psi)(X)
$$
where $\tilde{\theta}(x,x-y,x-z):=\partial_{x_j}\theta(x,x-y,x-z)$. These two formulas allow us to control all the seminorms in the corresponding H\"ormander symbol spaces. 
\end{proof}

\vspace{0.5cm}

Fix $x,y,z\in\mathbb{R}
^d$ and define the following objects:
\begin{equation}\label{II.1.2}
D^B_{jk}(x,y,z)\ :=\
\int_0^1ds\int_0^sdt\,B_{jk}\big(x+(1-2s)y+(2t-1)z\big),
\end{equation}
\begin{equation}\label{II.1.1}
F_{B}(x,y,z):=\frac{1}{4}\int_{T(x,y,z)}\hspace*{-0.7cm}
B\hspace*{0.2cm}=\ \left\langle
D^B(x,y,z)\,z\,,\,y\right\rangle\ :=\ \underset{j\ne
k}{\sum}y_jz_kD^B_{jk}(x,y,z).
\end{equation}

\begin{corollary}\label{L.II.1.0}
For every $(m,m')\in\mathbb{R}\times\mathbb{R}$ and $\rho\in[0,1]$, and for any
magnetic field $B$ satisfying Hypothesis \ref{Hyp-B}, we have that the map
$
S^m_\rho(\mathbb{R}^d)\,\times\,S^{m^\prime}_\rho(\mathbb{R}^d)\,\ni\,(a,b)\,
\mapsto\ ,
a\sharp^{B}b\in S^{m+m^\prime}_\rho(\mathbb{R}^d)
$
is bilinear and continuous.
\end{corollary}
\begin{proof}
 We use the formula of the magnetic composition \eqref{II.1.5} and apply Proposition \ref{comp-symb} by replacing $\theta(x,y,z)$ with $e^{-4iF_B(x,y,z)}$.
\end{proof}

\vspace{0.5cm}

Using Theorem 3.1 and Remark 3.2 in \cite{IMP1} we deduce that for any $\Phi\in\mathscr{S}(\Xi)$ the operator
$\mathfrak{Op}^A(\Phi)$ defines a bounded operator on $L^2(\X)$. This allows us to
define on $\mathscr{S}(\Xi)$ a $C^*$-norm that only depends on $B$:
\beq\label{iulie4}
\|\Phi\|_{W,B}:=\|\mathfrak{Op}^A(\Phi)\|_{\mathbb{B}(L^2(\X))},
\qquad\forall\Phi\in\mathscr{S}(\Xi).
\eeq
Using the above cited results in \cite{IMP1} we can
extend the above $C^*$-norm to the space $S^0_0(\X)$. Moreover, it is also
proved in $\cite{IMP1}$ that  this norm is bounded by a specific norm from the
family defining the Fr\'{e}chet topology of $S^0_0(\X)$, that depends only on
the dimension $d=\dim_\mathbb{R}\X$.

\subsection{Weak magnetic fields}\label{SS-wmf}

In our paper we are interested in {\it weak} magnetic fields, that we shall
control by a small parameter $\epsilon\in[0,\epsilon_0]$.
\begin{hypothesis}\label{Hyp-V-magn}
We shall consider a family of magnetic fields
$\{B_\epsilon\}_{\epsilon\in[0,\epsilon_0]}$ of the form
$
B_\epsilon:=\epsilon\,B^0_\epsilon
$, with $B^0_\epsilon$ having components in $
BC^\infty\big(\X\big)
$ uniformly with respect to $\epsilon\in[0,\epsilon_0]$. In order to simplify the notation, when dealing with weak
magnetic fields, the indexes (or the exponents) $A_\epsilon$ or $B_\epsilon$ shall be replaced by $\epsilon$. Similarly, we shall
use the notation $\|\cdot\|_{W,\epsilon}$ instead of $\|\cdot\|_{W,B_\epsilon}$.
\end{hypothesis}

\vspace{0.5cm}

\begin{proposition}\label{L.II.1.1}
For $\epsilon\in[0,\epsilon_0]$ there exists a continuous
application
$r_\epsilon:S^m_\rho(\X)\times
S^{m^\prime}_\rho(\X)\rightarrow
S^{m+m^\prime-2\rho}_\rho(\X)$ such that:
\begin{equation}\label{II.1.3}
a\sharp^{\epsilon}b\,=\,a\sharp^{0}b\,+\,\epsilon
r_\epsilon(a,b),\qquad\forall(a,b)\in
S^m_\rho(\mathbb{R}^d)\times S^{m^\prime}_\rho(\mathbb{R}^d),
\end{equation}
\end{proposition}
\begin{proof}
In \eqref{II.1.5} we use the following
identity:
\begin{equation}\label{exp-dev}
e^{-4iF_{\epsilon}}\ =\
1\,-\,4iF_{\epsilon}\int_0^1e^{-4itF_{\epsilon}}dt,
\end{equation}
in order to obtain \eqref{II.1.3} with
\begin{align}\label{II.1.6}
&\big[r_\epsilon(a,b)\big](X)=\frac{1}{\epsilon}(a\sharp^{\epsilon}b\,-\,a\sharp^{0}b)\\
&=\,-\frac{4i}{(2\pi)^{2d}}\int_{\mathbb{R}^{4d}}e^
{-2i\sigma^\circ(Y,Z)}\left(\int_0^1e^{-4itF_{\epsilon}(x,y,z)}
dt\right)\left\langle
D^\epsilon(x,y,z)z,y\right\rangle a(X-Y)\,b(X-Z)\,dY\,dZ.\nonumber
\end{align}
All the components of the matrix $D^\epsilon$ belong to $
BC^\infty\big(\X\big)$ uniformly with respect to
$\epsilon\in[0,\epsilon_0]$ for any $j<k$. Integrating by parts we obtain:
\begin{align}\label{II.1.6.a}
&\big[r_\epsilon(a,b)\big](X)=\,-\frac{4i}{\pi^{2d}}\int_{\mathbb{R}^{4d}}\,dY\,dZe^
{-2i\sigma^\circ(Y,Z)}\left(\int_0^1e^{-4itF_{\epsilon}(x,y,z)}
dt\right)\\
&\times \sum_{j,k=1}^d D_{jk}^\epsilon(x,y,z)\big(\partial_{\xi_j}
a\big)(X-Y)\,\big(\partial_{\xi_k}
b\big)(X-Z).\nonumber 
\end{align}

The proof can be completed by using Proposition
\ref{comp-symb} where $$\theta(x,y,z)=\left(\int_0^1e^{-4itF_{\epsilon}(x,y,z)}
dt\right)D_{jk}^\epsilon(x,y,z).$$
\end{proof}

\vspace{0.5cm}

\begin{remark}
If we replace \eqref{exp-dev} with the $N$'th order Taylor expansion of the exponential, we obtain that for any
$N\in\mathbb{N}^*$:
\beq\label{iulie5}
a\sharp^{\epsilon}b=a\sharp^{0}b+\underset{1\leq k\leq
N-1}{\sum}\epsilon^kc^{(k)}_\epsilon(a,b)+\epsilon^N\rho^{(N)}_\epsilon(a,b),
\eeq
with $ c^{(k)}_\epsilon(a,b)\in S^{m+m^\prime-2k\rho}_1(\X)$ and $\rho^{(N)}_\epsilon(a,b)\in S^{m+m^\prime-2N\rho}_1(\X)$ 
uniformly in $\epsilon\in[0,\epsilon_0]$.
\end{remark}

Associated to the series development of the symbol $r^\epsilon_\z(h)$ given in
Proposition \ref{rez-dev}, we shall also use the notations
\beq\label{rez-dev-1}
r^\epsilon_{\z,n}(h):=\underset{0\leq k\leq
n}{\sum}\epsilon^kr_k(h;\epsilon,\z)\in
S^{-m}_1(\X)_\Gamma;\qquad\widetilde{r^\epsilon}_{\z,n}(h):=\underset{n+1\leq
k}{\sum}\epsilon^kr_k(h;\epsilon,\z)\in S^{-m}_1(\X)_\Gamma.
\eeq
\begin{remark}\label{rem-est-rest}
 The remainder $\widetilde{r^\epsilon}_{\z,n}\in S^{-m}_1(\X)_\Gamma$ has the
following properties:
\begin{enumerate}
 \item
$\widetilde{r^\epsilon}_{\z,n}\,=\,\epsilon^{n+1}\widetilde{\widetilde{
r^\epsilon}}_{\z,n}$, where $\widetilde{\widetilde{
r^\epsilon}}_{\z,n}\in S^{-m}_1(\X)_\Gamma$ uniformly in
$\epsilon\in[0,\epsilon_0)$ for
some small enough $\epsilon_0>0$;
\item
$h\sharp^\epsilon\widetilde{r^\epsilon}_{\z,n}\,=\,\epsilon^{n+1}
h\sharp^\epsilon
r^0_\z\sharp^\epsilon\left(\sum\limits_{k=n+1}^{\infty}\epsilon^{k-n-1}
\big(-r_\epsilon(h,r^0_\z)^{\sharp^\epsilon k}\big)\right)$ and noticing that 
$h\sharp^\epsilon r^0_\z=h\sharp^0 r^0_\z\,+\,\epsilon r_\epsilon(h,r^0_\z)$ we
conclude that
$\|h\sharp^\epsilon\widetilde{r^\epsilon}_{\z,n}\|_{W,\epsilon}\leq
C\epsilon^{n+1}$ for some $C>0$ independent of $\epsilon\in[0,\epsilon_0)$ for
some small enough $\epsilon_0>0$.
\end{enumerate}
\end{remark}

\subsubsection{Proof of Proposition \ref{rez-dev}}

In this paragraph we shall use the results above in order to prove Proposition
\ref{rez-dev}. The first point clearly follows from the spectral stability
results proved in \cite{AMP,CP-1}, as briefly recalled in subsection
\ref{main-res}, before the statement of Proposition \ref{rez-dev}. 

For the last two statements we start from the continuity of the
application
$$\mathbb{C}\setminus\sigma(H^\epsilon)\ni\z\mapsto r^\epsilon_\z(h)\in
S^{-m}_1(\X)$$
for the
$\|\cdot\|_{W,\epsilon}$-topology, a consequence of basic spectral theory for
self-adjoint operators. In order to obtain a control in the Fr\'{e}chet
topology on $S^{-m}_1(\X)$ we recall some results from \cite{IMP2}.  Let us
recall the symbol
$\mathfrak{s}_m$ introduced in Subsection \ref{M-PsiDO} just before Proposition
\ref{comp-symb} and the space of ``linear'' symbols:
$$\forall X\in\Xi,\ \mathfrak{l}_X(Y):=\sigma^\circ(X,Y).$$
For any $X\in\Xi$ we can define the operators
$$
\mathfrak{ad}_X^\epsilon[\psi]
:=\mathfrak{l}_X\sharp^\epsilon\psi-\psi\sharp^\epsilon\mathfrak{l}_X,
\qquad\forall\psi\in\mathscr{S}^\prime(\Xi).
$$
Then Theorem 5.2 in \cite{IMP2} states that the Fr\'{e}chet topology on any 
space $S^{-m}_1(\X)$ (for any $m\in\mathbb{R}$) may be also defined by the
following family of seminorms:
\beq
S^{-m}_1(\X)\ni\psi\mapsto\left\|\mathfrak{s}_{m+q}\sharp^\epsilon\big(\mathfrak
{ad}_{u_1}^\epsilon\cdots\mathfrak{ad}_{u_p}^\epsilon\mathfrak{ad}_{\mu_1}
^\epsilon\cdots\mathfrak{ad}_{\mu_q}^\epsilon[\psi]\big)\right\|_{W,\epsilon}
\in\mathbb{R}_+
\eeq
indexed by a pair of natural numbers $(p,q)\in\mathbb{N}\times\mathbb{N}$ and
by two families of points $\{u_1,\ldots,u_p\}\subset\X$ and
$\{\mu_1,\ldots,\mu_q\}\subset\X^*$. A simple computation shows that for any
$\epsilon\in[0,\epsilon_0]$ and any $\z\notin\sigma(H^\epsilon)$ 
\beq
\mathfrak{ad}_X^\epsilon[r^\epsilon_\z(h)]
=-r^\epsilon_\z(h)\sharp^\epsilon\mathfrak{ad}
_X^\epsilon[h]\sharp^\epsilon r^\epsilon_\z(h).
\eeq
Using the resolvent equation:
\beq
r^\epsilon_\z(h)=r^\epsilon_i(h)+(i-\z)r^\epsilon_i(h)r^\epsilon_\z(h),
\eeq
and Propositions 3.6 and 3.7 from \cite{IMP2} we easily prove that the
applications:
\beq\label{z-cont}
K\ni\z\mapsto\mathfrak{s}_{m+q}\sharp^\epsilon\big(\mathfrak
{ad}_{u_1}^\epsilon\cdots\mathfrak{ad}_{u_p}^\epsilon\mathfrak{ad}_{\mu_1}
^\epsilon\cdots\mathfrak{ad}_{\mu_q}^\epsilon[r^\epsilon_\z(h)]\big)\in
S^0_0(\X)
\eeq
are well defined, bounded and uniformly continuous for the norm
$\|\cdot\|_{W,\epsilon}$ for any $\epsilon\in[0,\epsilon_0]$.

The second point follows by noticing that the result in Lemma \ref{L.II.1.1}
implies the equality
\begin{equation}\label{II.1.9}
1\,=\,(h-\z)\sharp^0r^0_\z(h)\,=\,(h-\z)\sharp^{\epsilon}
r^0_\z(h)+\epsilon r_\epsilon\big(h,r^0_\z(h)\big)
\end{equation}
with the family
$\{r_\epsilon\big(h,r^0_\z(h)\big)\}_{\epsilon\in[0,\epsilon_0]}$ being a
bounded
subset in $S^{-2}(\X)$. We conclude that for some $\epsilon_0>0$ small enough,
$1+\epsilon r_\epsilon(h,r^0_\z(h))$ defines an invertible magnetic
operator for any $\epsilon\in[0,\epsilon_0]$
and its inverse will have a symbol $s^\epsilon(\z)$ given as the limit of the
following norm convergent series:
\beq
s^\epsilon(\z):=\underset{n\in\mathbb{N}}{\sum}\big(-\epsilon
r_\epsilon(h,r^0_\z(h))\big)^{\sharp^\epsilon n}\ \in\
S^0_1(\X).
\eeq
This clearly gives us the development in point (2) of the Theorem with
\beq
r_n(h;\epsilon,\z):=(-1)^nr^0_\z(h)\sharp^\epsilon\big(r_\epsilon(h,
r^0_\z(h))\big)^{\sharp^\epsilon n}\in S^{-(m+2n)}_1(\X).
\eeq

In order to control the uniformity with respect to $\epsilon\in[0,\epsilon_0]$
of the continuity of the application in \eqref{z-cont} let us notice that
\beq
r^\epsilon_\z(h)-r^\epsilon_{\z'}
(h)=(\z'-\z)r^\epsilon_\z(h)\sharp^\epsilon r^\epsilon_{\z'}(h)
\eeq
and that for any $\z\in K$ the family of symbols
$\{r^\epsilon_\z(h)\}_{\epsilon\in[0,\epsilon_0]}$, for $\epsilon_0>0$ small enough, is a bounded set in
$S^{-m}(\X)$ due to point (2).

\subsection{Proof of Theorem \ref{T-A}}
We shall consider {\it an
isolated spectral band} as in Hypothesis
\ref{is-sp-band} with a {\it band projection} and a {\it band Hamiltonian} as
in Definition \ref{Is-sp-band-op}; they
describe the exact dynamics in the given `energy window' in the absence of a
magnetic field.
Now suppose that a magnetic field satisfying Hypothesis \ref{Hyp-V-magn}
(depending on the parameter $\epsilon\in[0,\epsilon_0]$) is switched on.

We introduce the following simplified notations:
\beq\label{M-band-proj}
\mathfrak{p}^\epsilon:=E_{I}(H^\epsilon),\quad 
\mathfrak{H}^\epsilon:=H^\epsilon
E_{I}(H^\epsilon).
\eeq
We have that
$\big[
\mathfrak{H}^\epsilon\,,\,\mathfrak{p}^\epsilon\big]=0$ and
\beq\label{iulie6}
\mathfrak{p}^\epsilon\ =\ \frac{i}{2\pi }\int_\mathcal{C}r^\epsilon_\z(h)d\z,\qquad
\mathfrak{H}^\epsilon\ =\ \frac{i}{2\pi}\int_\mathcal{C}\z r^\epsilon_\z(h)d\z
\eeq
with $\mathcal{C}$ the smooth closed curve in $\mathbb{C}$ defined just before Remark \ref{Prop-bandOp}
for $H$ replaced with $H^\epsilon$ with $\epsilon\in[0,\epsilon_0]$ and $\epsilon_0>0$ small enough.

\begin{proposition}\label{PO-1}
Under Hypothesis \ref{is-sp-band} and \ref{Hyp-V-magn}, and using \eqref{rez-dev-1}, we have that
\beq\label{SH}
\mathfrak{H}^\epsilon\,=\,\mathfrak{Op}^\epsilon\big(S_{\mathfrak{H}}\big)+\underset{1\leq k\leq n-1}{\sum}
\epsilon^k\mathfrak{Op}^\epsilon\big(v^\epsilon_k\big)
  +\epsilon^n\mathfrak{Op}^\epsilon\big(R^\epsilon_I(h;n)\big),
\eeq
with
\beq
S_{\mathfrak{H}}(z,\zeta)\ =\ \underset{k\in
J_I}{\sum}\left(\int_{\X}e^{i<\zeta,v>}\int_{\mathbb{T}_*}
\lambda_k(\theta)\phi_k(z+v/2,\theta)
\overline{\phi_k(z-v/2,\theta)}\,d\theta\right)dv
\eeq
and
$$v^\epsilon_k:=-(2\pi i)^{-1}\int_{\mathcal{C}}\z r_k(h;\epsilon,\z)d\z,\quad R^\epsilon_I(h;n):=-(2\pi i)^{-1}\int_{\mathcal{C}}\z\left(\underset{k\geq n}{\sum}\epsilon^kr_k(h;\epsilon,\z)\right)d\z.$$
\end{proposition}
\begin{proof}
 We use \eqref{iulie6} and the development of $r^\epsilon_\z(h)$ given in the second point of 
 Proposition \ref{rez-dev}. Formula \eqref{SH} follows from \eqref{Prop-bandOp} and \eqref{SH-0}.
\end{proof}
This Proposition proves point 1 of our Theorem \ref{T-A}.
For the proof of the second point let us consider the projection
$\mathfrak{p}^\epsilon=:\mathfrak{Op}^\epsilon\big(S_{\mathfrak{p}^\epsilon}
\big)$ and try to
approximate it by the magnetic quantization of the `free symbol'
$S_{\mathfrak{p}}$. A technical difficulty comes from the fact that now
$\mathfrak{Op}^\epsilon(S_{\mathfrak{p}})$ is no longer a projection (it is
idempotent only modulo an error of order $\epsilon$!). 
\begin{proposition}\label{PO-1-proj}
 Under Hypothesis \ref{is-sp-band} and \ref{Hyp-V-magn}, with the above
notations, for any $\epsilon\in[0,\epsilon_0]$
with $\epsilon_0>0$ small enough and for any $N\in\mathbb{N}^*$ there exists
an orthogonal projection
$\widetilde{\mathfrak{p}}^\epsilon_N$ such that:
\begin{enumerate}
\item $\widetilde{\mathfrak{p}}^\epsilon_N\ =\ \underset{0\leq m\leq
N-1}{\sum}\epsilon^m\mathfrak{Op}^\epsilon\left(\frac{i}{2\pi }\int_{\mathcal{C}}
r_m(h;\epsilon,\z)\,d\z\right)\,+
\,\epsilon^N X^\epsilon_N(h)$, with $\|X^\epsilon_N(h)\|\leq C_N<\infty$ for any
$\epsilon\in[0,\epsilon_0]$. 
 \item
$\left\|\mathfrak{p}^\epsilon\,-\,\widetilde{\mathfrak{p}}^\epsilon_N\right\|\
\leq\,C_N(h)\epsilon^N$ for some $C_N(h)<\infty$.
 \item
$\left\|\mathfrak{H}^\epsilon\,-\,H^\epsilon\widetilde{
\mathfrak{p}}^\epsilon_N\right\|\
\leq\,C_N(h)\epsilon^N$, $\left\|\mathfrak{H}^\epsilon\,-\,\mathfrak{H}
^\epsilon\widetilde{
\mathfrak{p}}^\epsilon_N\right\|\ \leq\,C_N(h)\epsilon^N$ for some
$C_N(h)<\infty$.
\item 
$\left\|\big[\mathfrak{H}^\epsilon\,,\,\widetilde{\mathfrak{p}}^\epsilon_N\big]
\right\|\leq
C_N\epsilon^{N},\ \left\|\big[H^\epsilon\,,\,\widetilde{\mathfrak{p}}
^\epsilon_N\big]
\right\|\leq C_N\epsilon^{N},\
\ \forall\epsilon\in[0,\epsilon_0)$
\item
$\exists
R_{N}(h)\in S^{-\infty}(\X)$ such that
$H^\epsilon\widetilde{\mathfrak{p}}^\epsilon_N=\mathfrak{H}^\epsilon+
\epsilon^{N}\mathfrak{Op}^\epsilon(R_{N}(h))$.
\end{enumerate}
\end{proposition}
\begin{proof}
Using once again \eqref{iulie6} and the development of $r^\epsilon_\z(h)$ given in the second point of 
 Proposition \ref{rez-dev} we can write
\beq\label{p-epsilon-0}
S_{\mathfrak{p}^\epsilon}\ =\
\underset{m\in\mathbb{N}}{\sum}\epsilon^mS_{\mathfrak{p}^\epsilon,m},\qquad
S_{\mathfrak{p}^\epsilon,m}:=\frac{i}{2\pi
}\int_{\mathcal{C}}r_m(h,\epsilon,\z)\,d\z.
\eeq
We define for any $N\in\mathbb{N}^*$:
\beq
S_{\mathfrak{p}^\epsilon}^{(N)}:=-\frac{1}{2\pi
i}\int_{\mathcal{C}}r^\epsilon_{\z,N}(h)\,d\z,\qquad
Z^{(N)}_\epsilon:=S_{\mathfrak{p}^\epsilon}-S_{\mathfrak{p}^\epsilon}^{(N)}
\eeq
\beq
\overset{\circ}{S}{}_{\mathfrak{p}^\epsilon}^{(N)}:=(1/2)\Big(S_{\mathfrak{p}
^\epsilon}^{(N)}+\overline{S_{\mathfrak{p}^\epsilon}^{(N)}}\Big)\qquad
\overset{\circ}{Z}{}^{(N)}_\epsilon:=(1/2)\Big(Z^{(N)}_\epsilon+\overline{Z^{(N)
}_\epsilon}\Big)=S_{\mathfrak{p}^\epsilon}-\overset{\circ}{S}{}_{\mathfrak
{p}^\epsilon}^{(N)}.
\eeq
Then
$$
Z^{(N)}_\epsilon=\frac{i}{2\pi
}\int_{\mathcal{C}}\widetilde{r^\epsilon}_{\z,N}(h)\,d\z\in
S^{-m}_1(\X),\qquad\left\|\mathfrak{Op}^\epsilon\big(Z^{(N)}
_\epsilon\big)\right\|\leq C_N\epsilon^{N+1}.
$$

We define
\beq
\mathfrak{q}^\epsilon_N\ :=\
\mathfrak{Op}^\epsilon\big(\overset{\circ}{S}{}_{\mathfrak{p}^\epsilon}^{(N)}
\big),
\eeq
and notice that it is a self-adjoint operator and 
\begin{equation*}
\big(\overset{\circ}{S}{}_{\mathfrak{p}^\epsilon}^{(N)}\big)^{
\sharp^\epsilon2}\,-\,\overset{\circ}{S}{}_{\mathfrak{p}^\epsilon}^{(N)}
\,=\,\big(S_{\mathfrak{p}^\epsilon}-Z^{(N)}_\epsilon\big)^{
\sharp^\epsilon2}\,-\,\big(S_{\mathfrak{p}^\epsilon}-Z^{(N)}_\epsilon\big)\,=\,
Z^{(N)}_\epsilon\,-\,\big(Z^{(N)}_\epsilon
\sharp^\epsilon S_{\mathfrak{p}^\epsilon}+
S_{\mathfrak{p}^\epsilon}\sharp^\epsilon Z^{(N)}_\epsilon\big)\,+\,\big(
Z^{(N)}_\epsilon\big)^{\sharp^\epsilon2}
\end{equation*}
implying the estimate
$\left\|(\mathfrak{q}^\epsilon_N)^2-\mathfrak{q}
^\epsilon_N\right\|\leq C\epsilon^{N+1}$.
Following the
procedure in \cite{Ne93} (Proposition 3) we notice that this estimation
implies that there exists $\epsilon_1\in(0,\epsilon_0]$ small enough
such that $\sigma(\mathfrak{q}^\epsilon)=I_0\cup I_1$ where
$I_0\subset[-\epsilon^{N+1},\epsilon^{N+1}]$ and
$I_1\subset[1-\epsilon^{N+1},1+\epsilon^{N+1}]$ for any
$\epsilon\in[0,\epsilon_1]$. Thus
$(\epsilon^{N+1},1-\epsilon^{N+1})\notin\sigma(\mathfrak{q}^\epsilon_N)$ for any
$\epsilon\in[0,\epsilon_1]$. We can thus find a smooth contour
$\mathcal{C}_1\subset\mathbb{C}$ such that $[1-\epsilon^{N+1},1+\epsilon^{N+1}]$
be in in the interior region and $[-\epsilon^{N+1},\epsilon^{N+1}]$ in its
exterior region, and define
\beq\label{def-p-epsilon-m}
\widetilde{\mathfrak{p}}^\epsilon_N\,:=\,\frac{i}{2\pi
}\int_{\mathcal{C}_1}\big(\mathfrak{q}^\epsilon_{N-1}-
\z\big)^{-1}d\z.
\eeq
By definition, $\widetilde{\mathfrak{p}}^\epsilon_N$ is a self-adjoint projection
that commutes with $\mathfrak{q}^\epsilon_{N-1}$;
moreover it is equal to $E_{I_1}(\mathfrak{q}^\epsilon_{N-1})$. This means
that
$$
\widetilde{\mathfrak{p}}^\epsilon_N-\mathfrak{q}^\epsilon_{N-1}\,=\,\widetilde{
\mathfrak{p}}^\epsilon_N-\mathfrak{q}^\epsilon_{N-1}\big(\widetilde{\mathfrak{p}
}^\epsilon_N+(\bb1-\widetilde{\mathfrak{p}}^\epsilon_N)\big)\,=\,
(\bb1-\mathfrak{q}^\epsilon_{N-1})\widetilde{\mathfrak{p}}^\epsilon_N\,+\,
\mathfrak{q}^\epsilon_{N-1} (\bb1-\widetilde{\mathfrak{p}}^\epsilon_N)
$$ 
satisfies the estimation
$\|\widetilde{\mathfrak{p}}^\epsilon_N-\mathfrak{q}^\epsilon_{N-1}\|\leq
C\epsilon^{N}$.

From \eqref{def-p-epsilon-m} it follows that
\begin{align*}
S_{\widetilde{\mathfrak{p}}^\epsilon_N}\,&=\,-\frac{1}{2\pi
i}\int_{\mathcal{C}_1}\big(S^{(N-1)}_{\mathfrak{p}^\epsilon}-
\z\big)^{-}_\epsilon d\z\\
&=-\,\frac{1}{2\pi
i}\int_{\mathcal{C}_1}\big(S_{\mathfrak{p}^\epsilon}-
\z\big)^{-}_\epsilon d\z\,+\,\frac{1}{2\pi
i}\int_{\mathcal{C}_1}\big(S^{(N-1)}_{\mathfrak{p}^\epsilon}-
\z\big)^{-}_\epsilon \sharp^\epsilon\big(S^{(N-1)}_{\mathfrak{p}^\epsilon}-
S_{\mathfrak{p}^\epsilon}\big)\sharp^\epsilon\big(S_{\mathfrak{p}^\epsilon}-
\z\big)^{-}_\epsilon d\z\\
&=\,S_{\mathfrak{p}^\epsilon}\,+\,\frac{1}{2\pi
i}\int_{\mathcal{C}_1}\big(S^{(N-1)}_{\mathfrak{p}^\epsilon}-
\z\big)^{-}_\epsilon \sharp^\epsilon\left(\frac{1}{2\pi
i}\int_{\mathcal{C}}\widetilde{r^\epsilon}_{\z',N-1}(h)\,
d\z'\right)\sharp^\epsilon
\big(S_{\mathfrak{p}^\epsilon}-\z\big)^{-}_\epsilon d\z\\
&=\,S_{\mathfrak{p}^\epsilon}\,+\,\frac{1}{2\pi
i}\int_{\mathcal{C}_1}\big(S_{\mathfrak{p}^\epsilon}-\z\big)^{-}_\epsilon 
\sharp^\epsilon\left(\frac{1}{2\pi
i}\int_{\mathcal{C}}\widetilde{r^\epsilon}_{\z',N-1}(h)\,
d\z'\right)\sharp^\epsilon\big(S^{(N-1)}_{\mathfrak{p}^\epsilon}-
\z\big)^{-}_\epsilon d\z.
\end{align*}
Using the second point in Remark \ref{rem-est-rest} we easily obtain
the following stronger result:
\beq\label{L-p-epsilon-m}
\exists\epsilon_0>0\,\forall N\in\mathbb{N}^*\,\exists C_N<\infty\qquad
\left\|h\sharp^\epsilon\big(S_{\widetilde{\mathfrak{p}}^\epsilon_N}-S_{\mathfrak
{p}^\epsilon}\big)\right\|_{W,\epsilon}\ \leq\
C_N\epsilon^{N},\qquad\forall\epsilon\in[0,\epsilon_0).
\eeq
We also note that
$
\big[h\,,
\,S_{\widetilde{\mathfrak{p}}^\epsilon_N}\big]_\epsilon=h\sharp^\epsilon\big(S_{\widetilde{\mathfrak{p}}^\epsilon_N}-S_{\mathfrak{p}
^\epsilon}
\big)-\big(S_{\widetilde{\mathfrak{p}}^\epsilon_N}-S_{\mathfrak{p}^\epsilon}
\big)\sharp^\epsilon h.
$
From the bounds:
\beq
\left\|S_{\widetilde{\mathfrak{p}}^\epsilon_N}\,-\,S_{\mathfrak{p}_\epsilon}^{
(N-1)}\right\|_{W,\epsilon}\ =\
\left\|\widetilde{\mathfrak{p}}^\epsilon_N\,-\,\mathfrak{q}^\epsilon_{
N-1}\right\|\ \leq C\epsilon^{N}
\eeq
we conclude that the following estimates hold:
\beq
\left\|h\sharp^\epsilon\big(S_{\widetilde{\mathfrak{p}}^\epsilon_N}-S_{\mathfrak
{p}^\epsilon
}\big)\right\|_{W,\epsilon}\ \leq\
C_N\epsilon^{N},\qquad\forall\epsilon\in[0,\epsilon_0).
\eeq
Since
\beq
\left\|\mathfrak{q}^\epsilon_N\,-\,\mathfrak{p}^\epsilon\right\|\ \leq\ \frac{1}{2\pi
}\int_{\mathcal{C}}\|\widetilde{r^\epsilon}_{\z,N}\|_{W,\epsilon}\,|d\z|
 \leq\ C\epsilon^{N+1},
\eeq
we can write:
\beq\label{iulie7}
\left\|S_{\widetilde{\mathfrak{p}}^\epsilon_N}-S_{\mathfrak{p}^\epsilon}
\right\|_{W,\epsilon}\
\leq\ C\epsilon^{N},  \quad \left\|h\sharp^\epsilon\big(S_{\widetilde{\mathfrak{p}}^\epsilon_N}-S_{\mathfrak
{p}^\epsilon}\big)\right\|_{W,\epsilon}\ \leq\
C\epsilon^{N}.
\eeq
Using the magnetic symbolic calculus \cite{IMP1} we end the proof of the
Proposition.
\end{proof}

\vspace{0.5cm}

\begin{corollary}\label{PO-1-ham}
Under Hypothesis \ref{is-sp-band} and \ref{Hyp-V-magn}, with the above
notations,  for any $N\in\mathbb{N}^*$ there exists $\epsilon_0>0$ such that for any
$\epsilon\in[0,\epsilon_0]$  we have:
\begin{align}\label{T.2}
e^{-itH^\epsilon}E_{I}(H^\epsilon)\,&=\,e^{-it\mathfrak{H}
^\epsilon}\widetilde{\mathfrak{p}}
^\epsilon_N\,+\,\epsilon^N{\mathfrak{R}'}_N^\epsilon(h)\,
\end{align}
with 
$\|{\widetilde{\mathfrak{R}'}_N}^\epsilon(h)\|\leq C_N |t|$ for some $C_N<\infty$ and any
$\epsilon\in[0,\epsilon_0]$.
\end{corollary}
\begin{proof}
In order to prove \eqref{T.2} let us define the application:
$$
\mathbb{R}\ni t\mapsto\Psi(t):=e^{-itH^\epsilon}\widetilde{\mathfrak{p}}
^\epsilon_N
e^{it\mathfrak{H}^\epsilon}\in\mathbb{B}\big(L^2(\X)\big)
$$
which due to \eqref{iulie7} is norm differentiable for all $t\in\mathbb{R}$ and has the derivative:
$$
\dot{\Psi}(t)=-ie^{-itH^\epsilon}\left(H^\epsilon\widetilde{\mathfrak{p}}
^\epsilon_N-
\widetilde{\mathfrak{p}}^\epsilon_N\mathfrak{H}^\epsilon
\right)e^{it\mathfrak{H}^\epsilon}\in\mathbb{B}\big(L^2(\X)\big).
$$
From \eqref{iulie7} we get
$
\left\|H^\epsilon\widetilde{\mathfrak{p}}^\epsilon_N-\widetilde{\mathfrak{p}}
^\epsilon_N\mathfrak{H}
^\epsilon\right\|_{\mathbb{B}(L^2(\X))}
\leq C_N\epsilon^N$ and also
$$
\left\|e^{-itH^\epsilon}\widetilde{\mathfrak{p}}^\epsilon_N-\widetilde{\mathfrak
{p}}^\epsilon_N
e^{-it\mathfrak{H}^\epsilon}\right\|_{\mathbb{B}(L^2(\X))}\leq\left |\int_0^{t}\left\|
\dot{\Psi}(s)\right\|_{\mathbb{B}(L^2(\X))}ds\right |\leq C_N(h)|t|\epsilon^N.
$$
If we define
$$
\Phi(t):=e^{-it\mathfrak{H}^\epsilon}\widetilde{\mathfrak{p}}^\epsilon_N-
\widetilde{\mathfrak{p}}^\epsilon_N e^{-it\mathfrak{H}^\epsilon}
$$
and an argument similar to the one above, using now the estimation for
$\big[\mathfrak{H}^\epsilon\,,\,\widetilde{\mathfrak{p}}^\epsilon_N\big]$ allows
us to conclude that
$$
\left\|e^{-it\mathfrak{H}^\epsilon}\widetilde{\mathfrak{p}}
^\epsilon_N-\widetilde{\mathfrak{p}}^\epsilon_N
e^{-it\mathfrak{H}^\epsilon}\right\|_{\mathbb{B}(L^2(\X))}\leq
C_N(h)|t|\epsilon^N.
$$
Hence:
$$
\left\|e^{-itH^\epsilon}\widetilde{\mathfrak{p}}^\epsilon_N-e^{-it\mathfrak{H}^\epsilon}\widetilde{\mathfrak
{p}}^\epsilon_N
\right\|_{\mathbb{B}(L^2(\X))}\leq C_N(h)|t|\epsilon^N,
$$
which together with $|| \widetilde{\mathfrak
{p}}^\epsilon_N-\mathfrak
{p}^\epsilon||\leq C_N(h)\epsilon^N$ it ends the proof of the Proposition.
\end{proof}

\subsection{Proof of Theorem \ref{teoremadoi}}

Having the composite Wannier basis at our disposal (see \eqref{Wannier-def}), we shall 
construct a {\it `magnetic' Wannier basis} for the magnetic band projection
$E_I(H^\epsilon)$ as in  \cite{CHN}, based on the strategy in \cite{Ne-RMP}. 

The integral kernel corresponding to
$E_I(H)$ may be written as
$$K(\mathfrak{p})(x,y)\ =\ \underset{\gamma\in\Gamma}{\sum}\left(\underset{j\in
J_I}{\sum}w_j(x-\gamma)\overline{w_j(y-\gamma)}\right)$$
where the series is absolutely convergent due to the spatial localization of the Wannier functions (see Proposition \ref{iulie30}).

\begin{definition}\label{iulie9}
We introduce the following modified band projection kernel
\beq\label{KepsilonB}
K_\epsilon(\mathfrak{p})(x,y)\ :=\ \underset{\gamma\in\Gamma}{\sum}\left(\underset{j\in
J_I}{\sum}\omega_\gamma^\epsilon(x,y)w_j(x-\gamma)\overline{w_j(y-\gamma)}\right)
,\qquad\omega_\gamma^\epsilon(x,y)\ :=\ e^{i\int_{<\gamma,x,y>}B_\epsilon}
\eeq
and the following bounded self-adjoint operators on $L^2(\X)$:
\begin{enumerate}
 \item
$\mathfrak{q}^\epsilon\,:=\,\mathcal{I}\text{\sf
nt}[\widetilde{\Lambda}^\epsilon K(\mathfrak{p})]=\mathfrak
{Op}^\epsilon\left(S_{\mathfrak{p}}\right)$,
\item $\widetilde{\mathfrak{q}}^\epsilon\,:=\mathcal{I}\text{\sf
nt}[\widetilde{\Lambda}^\epsilon K_\epsilon(\mathfrak{p})]$.
\end{enumerate}
\end{definition}
We remark that 
\begin{align}\label{iulie8}
\widetilde{\Lambda}^\epsilon(x,y)K_\epsilon(\mathfrak{p})(x,y)=
\underset{\gamma\in\Gamma}{\sum}\left(\underset{j\in
J_I}{\sum}\widetilde{\Lambda}^\epsilon(x,\gamma)w_j(x-\gamma)\overline{
\widetilde{\Lambda}^\epsilon(y,\gamma)w_j(y-\gamma)}\right).
\end{align}
\begin{definition}\label{modif-W-func}
Let us define the following {\it modified Wannier} functions
(similar to those used in
\cite{HS1,Ne-RMP, CHN} for the constant magnetic field case)
$$
\widetilde{\mathcal{W}}^\epsilon_{\gamma,j}(x)\ :=\
\widetilde{\Lambda}^\epsilon(x,\gamma)w_j(x-\gamma).
$$
\end{definition}

\begin{lemma}\label{band-est}
 For any $m\in\mathbb{N}$ there exists a constant $C_m<\infty$
such that
$$
\underset{(x,y)\in\X\times\X}{\sup}<x-y>^{m}
\left|K_\epsilon(\mathfrak{p})(x,y)\,-\,K(\mathfrak{p})(x,y)\right|\ \leq
C_m\epsilon.
$$
\end{lemma}
\begin{proof}
First we fix some $\gamma\in\Gamma$ and notice that for any $(N_1,N_2)\in\mathbb{N}^2$
there exists a constant $C<\infty$ such that
\beq
\underset{(x,y)\in\X\times\X}{\sup}<x-\gamma>^{N_1}<y-\gamma>^{N_2}
\left|\omega_\gamma^\epsilon(x,
y)-1\right||w_j(x-\gamma)||w_j(y-\gamma)|\ \leq C\epsilon.
\eeq
This allows us to conclude that for any $m\in\mathbb{N}$ and any
$p\in\mathbb{N}$
\begin{align*}
<x-y>^{m}
\left|K_\epsilon(\mathfrak{p})(x,y)\,-\,K(\mathfrak{p})(x,y)\right|\  
\leq\ 
\epsilon
C_{m,p}\underset{(x,y)\in\X\times\X}{\sup}\;\underset{\gamma\in\Gamma}{\sum}<x-\gamma>^{
-p}<y-\gamma>^{-p},
\end{align*}
which ends the proof.
\end{proof}

\vspace{0.5cm}

\begin{proposition}\label{tildaq-p}
There exists some $C>0$
such that
$$
\left\|\widetilde{\mathfrak{q}}^\epsilon\,-\,\mathfrak{p}^\epsilon\right\|\,
\leq\,\epsilon C,\qquad\forall\epsilon\in[0,\epsilon_0].
$$
\end{proposition}
\begin{proof}
If we replace $N=1$ in Proposition \ref{PO-1-proj} (1) and (2), and using Definition \ref{iulie9} we obtain
\beq
\|\mathfrak{p}^\epsilon-
\mathfrak{q}^\epsilon\|\leq C\; \epsilon.
\eeq
From our Lemma \ref{band-est} we get that
$\left\|\widetilde{\mathfrak{q}}^\epsilon\,-\,\mathfrak{q}^\epsilon\right\|\,
\leq\,\epsilon C$ and the triangle inequality finishes the proof.
\end{proof}

\vspace{0.5cm}

Since 
$\widetilde{\Lambda}^\epsilon(Q,\gamma)\tau_{-\gamma}$ is
unitary, for every fixed $\gamma\in\Gamma$ the set $\{\widetilde{W}^\epsilon_{\gamma,j}\}_{j\in J_I}$
 is orthonormal, but this is no longer true when
considering pairs
$(\widetilde{\mathcal{W}}^\epsilon_{\alpha,j},\widetilde{\mathcal{W}}^\epsilon_{\beta,k})$
with $\alpha\ne\beta$. We shall apply an orthogonalization procedure in order to obtain an orthonormal basis for the range of $\widetilde{\mathfrak{q}}^\epsilon$.
\begin{lemma}\label{Gramm-Wannier}
 With the above notations, by writing 
$\left\langle\widetilde{\mathcal{W}}^\epsilon_{\alpha,j},\widetilde{\mathcal{W}}
^\epsilon_{\beta,k}\right\rangle_{L^2(\X)}\,=\,\delta_{\alpha\beta}\delta_{jk}\,
+\,\epsilon\big[\mathbb{X}^\epsilon_{\alpha,\beta}\big]_{j,k}$,
then for any $m\in\mathbb{N}$ there exists $C_m>0$ such that  for any $\epsilon\in[0,\epsilon_0]$
$$|\alpha-\beta|^m\left|\big[\mathbb{X}^\epsilon_{\alpha,\beta}\big]_{j,k}
\right|\leq C_m.$$ 
\end{lemma}
\begin{proof}
\beq
\left\langle\widetilde{\mathcal{W}}^\epsilon_{\alpha,j},\widetilde{\mathcal{W}}
^\epsilon_{\beta,k}\right\rangle=
\int_\X\widetilde {\Lambda}
^\epsilon(\alpha,x)\overline{w_j(x-\alpha)}
\,\widetilde {\Lambda}
^\epsilon(x,\beta)w_k(x-\beta)\,dx\,=\,\delta_{\alpha\beta}\delta_{jk}\,+
\eeq
$$
+\,\widetilde
{\Lambda}^\epsilon(\alpha,\beta)\int_X\Omega^\epsilon(x,\alpha,\beta)
\overline{w_j(x-\alpha)}w_k(x-\beta)\,dx
$$
where for any $m\in\mathbb{N}$
$$
|\alpha-\beta|^m\left|\int_X\Omega^\epsilon(x,\alpha,\beta)
\overline{w_j(x-\alpha)}w_k(x-\beta)\,dx\right|\leq
$$
\beq
\leq\epsilon\,
C_\epsilon\int_\X<x-\alpha>^{
m+1}\left|w_j(x-\alpha)\right|<x-\beta>^{m+1}\left|w_k(x
-\beta)\right|\,dx\leq\epsilon\,C'_{m,\epsilon}
\eeq
for some $C'_{m,\epsilon}>0$ uniformly bounded for $\epsilon\in[0,\epsilon_0]$.
\end{proof}

\vspace{0.5cm}

We shall follow \cite{Ne-RMP} replacing the one-dimensional situation considered
there with our $N$-dimensional situation. Our family of modified Wannier functions can be seen as elements in $l^\infty(\Gamma;L^2(\X;\mathbb{C}^N))$. 

Lemma \ref{Gramm-Wannier} can be restated as 
$
\mathbb{G}^\epsilon\ =\ \bb1\ +\
\epsilon\mathbb{X}^\epsilon$, where $\big[\mathbb{G}^\epsilon_{\alpha,\beta}\big]_{j,k}\,:=\,\left\langle\widetilde{
\mathcal{W}}^\epsilon_{\alpha,j},\widetilde{\mathcal{W}}
^\epsilon_{\beta,k}\right\rangle_{L^2(\X)}$ and 
$\mathbb{G}^\epsilon_{\alpha,\beta}\in\mathbb{B}(\mathbb{C}^N)$ for
$(\alpha,\beta)\in\Gamma\times\Gamma$.
Both families $\mathbb{G}^\epsilon$ and $\mathbb{X}^\epsilon$ define bounded
operators both on $l^2\big(\Gamma;\mathbb{C}^N\big)$ and on
$l^\infty\big(\Gamma;\mathbb{C}^N\big)$, uniformly for
$\epsilon\in[0,\epsilon_0]$. Thus, as bounded operators on 
$l^2\big(\Gamma;\mathbb{C}^N\big)$, for $\epsilon_0>0$
small enough, we can define $\big(\mathbb{G}^\epsilon\big)^{-1/2}$ as
norm-convergent power series in $\epsilon\in[0,\epsilon_0]$. 

\begin{lemma}\label{G-lemma}
 We can write
\beq
\big(\mathbb{G}^\epsilon\big)^{-1/2}\ =\ \bb1\ +\ \epsilon\mathbb{Y}^\epsilon
\eeq
with $\mathbb{Y}^\epsilon_{\alpha,\beta}\in\mathbb{B}(\mathbb{C}^N)$ such that for any $m\in\mathbb{N}$ there exists $C_m>0$ such that for any $\epsilon\in[0,\epsilon_0]$:
$$|\alpha-\beta|^m\left\|\mathbb{Y}^\epsilon_{\alpha,\beta}
\right\|_{\mathbb{B}(\mathbb{C}^N)}\leq C_m.$$  
Thus
$\big(\mathbb{G}^\epsilon\big)^{-1/2}$ can be extended to a bounded operator on
$l^\infty\big(\Gamma;\mathbb{C}^N\big)$, uniformly for
$\epsilon\in[0,\epsilon_0]$.
\end{lemma}
\begin{proof}
We just have to use the norm convergent series expansion of the operator $\big(\mathbb{G}^\epsilon\big)^{-1/2}$
and notice that for each term of order $k\geq2$ we can write
$$
<\alpha-\beta>^m\leq c_m^{k-1}<\alpha-\gamma_1>^m<\gamma_1-\gamma_2>^m\ldots<\gamma_{k-2}-\gamma_{k-1}>^m<\gamma_{k-1}-\beta>^m.
$$
\end{proof}

 We define the family
\beq\label{on-msyst}
\overset{\circ}{\mathbb{W}}{}^\epsilon\ :=\
\big(\mathbb{G}^\epsilon\big)^{-1/2}\,\widetilde{\mathbb{W}}^\epsilon\in
 L^\infty\big(\Gamma;\big[L^2(\X)\big]^N\big).
\eeq
Its components
$\{\overset{\circ}{\mathcal{W}}{}^\epsilon_{\gamma,j}\}_{(\gamma,
j)\in\Gamma\times
J_I}$ form an
orthonormal family and
$\|\overset{\circ}{\mathcal{W}}{}^\epsilon_{\gamma,j}-\widetilde{\mathcal{W}}
^\epsilon_{\gamma,j}\|\leq C\epsilon$
uniformly in $(\gamma,j)\in\Gamma\times J_I$ and $\epsilon\in[0,\epsilon_0]$.
Moreover, for any $\gamma\in\Gamma$ and $m\in\mathbb{N}$ there exists
$C_m>0$ such that
\beq\label{approx-Wannier-f}
\underset{x\in\X}{\sup}<x-\gamma>^m\left|\mathcal{W}^\epsilon_{\gamma,j}(x)-\widetilde{\mathcal{W}}
^\epsilon_{\gamma,j}(x)\right|\leq C_m\epsilon
\eeq
uniformly in $(\gamma,j)\in\Gamma\times J_I$ and $\epsilon\in[0,\epsilon_0]$.
\begin{definition}\label{D-ptildaepsilon}
Let us define the orthogonal projection
$$
\widetilde{\mathfrak{p}}^\epsilon\ :=\
\underset{\gamma\in\Gamma}{\sum}\left(\underset{j\in
J_I}{\sum}|\overset{\circ}{\mathcal{W}}{}^\epsilon_{\gamma,j}
\rangle\langle\overset{\circ}{\mathcal{W}}{}^\epsilon_{\gamma,j}|\right).
$$
\end{definition}

\begin{proposition}~
\begin{itemize}
\item There exists $C>0$ such that
$\|\widetilde{\mathfrak{p}}^\epsilon\,-\,\mathfrak{p}
^\epsilon\|\leq\epsilon C$ for any
$\epsilon\in[0,\epsilon_0]$.
\item For any $\epsilon\in[0,\epsilon_0]$ the operator 
$
\mathscr{U}^\epsilon:=\left[\bb1\,-\,\big(\mathfrak{p}
^\epsilon\,-\,\widetilde{\mathfrak{p}}^\epsilon\big)^2\right]^{-1/2}
\left[\mathfrak{p}
^\epsilon\,\widetilde{\mathfrak{p}}^\epsilon
\,+\,\big(\bb1-\mathfrak{p}
^\epsilon\big)\big(\bb1-\widetilde{\mathfrak{p}}^\epsilon\big)\right]
$
is a unitary operator satisfying the intertwining property
$\mathscr{U}^\epsilon\widetilde{\mathfrak{p}}^\epsilon=\mathfrak{p}
^\epsilon\mathscr{U}^\epsilon$. Thus the family
$\{\mathscr{U}^\epsilon\overset{\circ}{\mathcal{W}}{}^\epsilon_{\gamma,j}\}_{
(\gamma,j)\in\Gamma\times
J_I}$ is an orthonormal basis of $\mathfrak{p}^\epsilon\mathcal{H}$ which is an
invariant subspace of $H^\epsilon$.
\item We have the following norm-convergent power series expansion:
$$
\mathscr{U}^\epsilon\ =\
\bb1\,+\,\underset{m\in\mathbb{N}^*}{\sum}\epsilon^m\mathbf{T}_m(\epsilon),
$$
where all the operators $\mathbf{T}_m(\epsilon)$ are uniformly bounded in
$\epsilon\in[0,\epsilon_0]$.
\end{itemize}
\end{proposition}
\begin{proof} 
We use \eqref{approx-Wannier-f} in order to 
estimate the operator norm of 
$\widetilde{\mathfrak{p}}^\epsilon\,-\,\widetilde{\mathfrak{q}}^\epsilon$ and show that it goes to zero like $\epsilon$. Then we use Proposition \ref{tildaq-p} and the Sz. Nagy formula  of the unitary intertwining operator (see
Remark II.4.4 in \cite{Ka}). 
\end{proof}

\vspace{0.5cm}

 We are now ready to define the {\it `magnetic Wannier' functions} as the following family
indexed by $(\gamma,j)\in\Gamma\times J_I$:
 \begin{equation*}
\mathcal{W}^\epsilon_{\gamma,j}\ :=\
\mathscr{U}^\epsilon\overset{\circ}{\mathcal{W}}{}^\epsilon_{\gamma,j},\quad \mathcal{W}^\epsilon_{\gamma,j}=E_I(H^\epsilon)\mathcal{W}^\epsilon_{\gamma,j}.
\end{equation*}

\subsubsection{Concluding the proof of Theorem \ref{teoremadoi}}

The first two points of the Theorem are direct consequences of the definitions
and arguments above from which we also conclude that for any $m\in\mathbb{N}$
there exists a finite positive constant $C_m$ such that
$$
\underset{(\alpha,\beta)\in\Gamma\times\Gamma}{\sup}\
<\alpha-\beta>^m\left|\left\langle\mathcal{W}^\epsilon_{\alpha,l},H
^\epsilon \mathcal{W}^\epsilon
_{\beta,m}\right\rangle_{L^2(\X)}\ -\
\left\langle\widetilde{\mathcal{W}}^\epsilon_{\alpha,l},\mathfrak{H}
^\epsilon\widetilde{\mathcal{W}}^\epsilon
_{\beta,m}\right\rangle_{L^2(\X)}\right|\ \leq\ C_m\epsilon
$$
which is independent of the indices
$(\alpha,\beta)\in\Gamma\times\Gamma$ and $(j,k)\in J_I\times J_I$.
Now we compute:
\begin{align*}
&\left\langle\widetilde{\mathcal{W}}^\epsilon_{\alpha,j},
\mathfrak{H}
^\epsilon\widetilde{\mathcal{W}}^\epsilon
_{\beta,k}\right\rangle_{L^2(\X)}\ \\
&=\ \int_\X\int_\X\overline{w_j(x-\alpha)}
\widetilde{\Lambda}^\epsilon(x,\alpha)\widetilde{\Lambda}^\epsilon(x,
x')\widetilde{\Lambda}^\epsilon(\beta,x')K_{\mathfrak{H}}(x,x')w_k(x'-\beta)dx\,
dx'\\
&=\
\widetilde{\Lambda}^\epsilon(\alpha,
\beta)\int_\X\int_\X\Omega^\epsilon(\alpha,x,x')\Omega^\epsilon(\beta,x,
x')K_{\mathfrak{H}}(x,x')\overline {w_j(x-\alpha)}w_k(x'-\beta)dx\,
dx'\\
&=\
\widetilde{\Lambda}^\epsilon(\alpha,\beta)\widehat{\mu_{jk}}_{\alpha-\beta}\ +\
(X_{j,k}^\epsilon)_{\alpha\beta}=\mathfrak{Op}^\epsilon_\Gamma(\widetilde{\mu}_{jk})_{\alpha,\beta}+(X_{j,k}^\epsilon)_{\alpha\beta}
\end{align*}
where in the last equality we used \eqref{iulie20}. Using Remark \ref{Omega} and the fast decay of the
Wannier functions $\{w_j(x)\}_{j\in J_I}$, we obtain for any $m\in\mathbb{N}$ the
estimate
\beq
\underset{(\alpha,\beta)\in\Gamma\times\Gamma}{\sup}\
<\alpha-\beta>^m\left|(X_{j,k}^\epsilon)_{\alpha\beta}\right|\ \leq\ C_m\epsilon
\eeq
uniformly in $(j,k)\in J_I\times J_I$ and $\epsilon\in[0,\epsilon_0]$. This implies that for every $m\geq 1$ there exists a constant $C_m>0$ such that 
\beq\label{iulie12}
\left |\left\langle
\mathcal{W}^\epsilon_{\alpha,j},H^\epsilon\mathcal{W}^\epsilon
_{\beta,k}\right\rangle_{L^2(\X)}\ -
\widetilde{\Lambda}^\epsilon(\alpha,\beta)\big(\widehat{\mu_{jk}}
\big)_{\alpha-\beta}\right | \leq C_m \epsilon <\alpha-\beta>^{-m}.
\eeq
In particular, this estimate proves the third point of Theorem \ref{teoremadoi} and the spectral result given in Corollary \ref{iulie33}.

\subsubsection{An alternative form for the effective Hamiltonian.}

We shall cast now the result in Theorem \ref{teoremadoi} in a form that generalizes formulas \eqref{PO-matrix} to the case of a small, smooth magnetic field. 

We consider the Hilbert space $\mathcal{H}_N:=L^2(\X)^N\cong\mathcal{H}\otimes\mathbb{C}^N$, with $N\in\mathbb{N}^*$ from Hypothesis \ref{is-sp-band}. We denote by $\bb1_N$ the identity operator on $\mathbb{C}^N$.

We notice that due to the decay properties of the 'magnetic' Wannier functions introduced above, using the Cotlar-Stein Lemma allows us to prove that for any pair $(j,k)\in J_I\times J_I$ the following series are strongly convergent in $\mathbb{B}(\mathcal{H})$ defining bounded operators (with good estimations on the norm):
\beq
\Pi^\epsilon_{jk}:=\underset{\gamma\in\Gamma}{\sum}\left|\mathcal{W}^\epsilon_{\gamma,j}\right\rangle\left\langle\mathcal{W}^
\epsilon_{\gamma,k}\right|.
\eeq
It is easy to see that $\Pi^\epsilon_{jj}$ are orthogonal projections for any $j\in J_I$ and we have the relations:
\beq
\big[\Pi^\epsilon_{jk}\big]^*=\Pi^\epsilon_{kj},\quad\big[\Pi^\epsilon_{jk}\big]^*\Pi^\epsilon_{jk}=\Pi^\epsilon_{jj},\quad \Pi^\epsilon_{jk}\big[\Pi^\epsilon_{jk}\big]^*=\Pi^\epsilon_{kk}.
\eeq
Let us define the operators $\widetilde{\Pi}^\epsilon\in\mathbb{B}(\mathcal{H}_N)$ given by the $\mathbb{B}(\mathcal{H})$-valued $N\times N$ matrix with entries $\Pi^\epsilon_{jk}$ and $\mathfrak{Op}^\epsilon(\widetilde{\mu})$ given by the $\mathbb{B}(\mathcal{H})$-valued $N\times N$ matrix with entries $\mathfrak{Op}^\epsilon(\widetilde{\mu}_{jk})$ for $(j,k)\in J_I\times J_I$.

\begin{proposition}
Under the Hypothesis of Theorem \ref{teoremadoi}, there exists some constant $C\in\mathbb{R}_+$ such that for some small enough $\epsilon_0>0$, we have for any $\epsilon\in[0,\epsilon_0]$ that:
$$
\left\|\mathfrak{H}^\epsilon\otimes\bb1_N\ -\ \widetilde{\Pi}^\epsilon\mathfrak{Op}^\epsilon(\widetilde{\mu})
\big[\widetilde{\Pi}^\epsilon\big]^*\right\|_{\mathbb{B}(\mathcal{H}_N)}\ \leq\ C\epsilon.
$$
\end{proposition}

This Proposition follows imediately from our Theorem \ref{teoremadoi} and the following statement.

\begin{proposition}\label{P-PO-W}
 For an isolated spectral band $I\subset\mathbb{R}$ that also satisfies
Hypothesis \ref{H.II.8.1} and for the linearly independent system
$\{\widetilde{\mathcal{W}}^\epsilon_{\alpha,j}\}_{(\alpha,j)\in\Gamma\times
J_I}$ defined above, the following relation holds:
$$
\left\langle\widetilde{\mathcal{W}}^\epsilon_{\alpha,l},\mathfrak{Op}
^\epsilon(\widetilde{\mu}_{lm})
\widetilde{\mathcal{W}}^\epsilon
_{\beta,m}\right\rangle_{L^2(\X)}\ =\
\widetilde{\Lambda}^\epsilon(\alpha,\beta)\big(\widehat{\mu_{lm}}\big)_{
\alpha-\beta}\,+\epsilon\,C_\epsilon
$$
for any $(\alpha,\beta)\in\Gamma\times\Gamma$ and any $(l,m)\in J_I\times
J_I$,
with some positive constant $C_\epsilon$ that is uniformly bounded for
$\epsilon\in[0,\epsilon_0]$.
\end{proposition}
\begin{proof}
 It is evidently enough to compute the following scalar products:
 \begin{align*}
 \left\langle\widetilde{\mathcal{W}}^\epsilon_{\alpha,l},\mathfrak{Op}
^\epsilon(\widetilde{
\mu}_{lm})\widetilde{\mathcal{W}}^\epsilon_{\beta,m}\right\rangle_{
L^2(\X)}&=\\
 =\underset{\gamma\in\Gamma}{\sum}\big(\widehat{\mu_{lm}}\big)_{-\gamma}&\int_\X
 \widetilde{\Lambda}^\epsilon(\alpha,x)\overline{\psi_l(x-\alpha)}
 \widetilde{\Lambda}^\epsilon(x,x+\gamma)\big(\tau_{\gamma}\widetilde{\Lambda}
^\epsilon(Q,\beta)\tau_{-\beta}\psi_m\big)(x)\,dx\\
 =\underset{\gamma\in\Gamma}{\sum}\big(\widehat{\mu_{lm}}\big)_{-\gamma}&\int_\X
 \widetilde{\Lambda}^\epsilon(\alpha,x)\overline{\psi_l(x-\alpha)}
 \widetilde{\Lambda}^\epsilon(x,x+\gamma)\widetilde{\Lambda}
^\epsilon(x+\gamma,\beta)\psi_m(x-\beta+\gamma)\,dx
\end{align*}
\begin{align*}
 =\widetilde{\Lambda}^\epsilon(\alpha,\beta)\big(\widehat{\mu_{lm}}\big)_{
\alpha-\beta}\,&+
 \\
 +\,\underset{\gamma\in\Gamma}{\sum}\big(\widehat{\mu_{lm}}\big)_{
-\gamma}&\int_\X\left[\Omega^\epsilon(\alpha,x,
x+\gamma)\Omega^\epsilon(x,x+\gamma,\beta)\,-\,1\right]
\overline{\psi_l(x-\alpha)}\psi_m(x-\beta+\gamma)\,dx.
 \end{align*}
 In order to estimate the last contribution above we notice that
 $$
 \underset{\gamma\in\Gamma}{\sum}\left|\big(\widehat{\mu_{lm}}\big)_{
-\gamma}\right|\int_\X\left|\Omega^\epsilon(\alpha,x,
x+\gamma)\Omega^\epsilon(x,x+\gamma,\beta)\,-\,1\right|
|\overline{\psi_l(x-\alpha)}|\,|\psi_m(x-\beta+\gamma)|\,dx\leq
 $$
 $$
 \leq\epsilon\,C_\epsilon\underset{\gamma\in\Gamma}{\sum}\left|\big(\widehat{
\mu_{lm}}\big)_{
-\gamma}\right|<\gamma>^2\int_\X<x-\alpha>|\psi_l(x-\alpha)|\,<x-\beta+\gamma>
|\psi_m(x-\beta+\gamma)|\,dx\leq\epsilon\,C'_\epsilon.
 $$
\end{proof}

\subsection{Proof of Theorem \ref{teorema3}}

Let us recall that when composite Wannier functions exist, the non-magnetic band Hamiltonian $HE_I(H)$ is 
unitarily equivalent with a matrix in $l^2(\Gamma)^N$ (see Proposition \ref{Band-operators}) given by 
$\widehat{\mu_{jk}}_{\alpha-\beta}$, where $\alpha,\beta\in \Gamma$ and $1\leq j,k\leq N$.

We now investigate what happens when the magnetic field perturbation $B_\epsilon$ is not just globally small but it also has slow variation, i.e. it is generated by a vector potential $A_\epsilon(x)=A(\epsilon x)$, where $A$ has bounded first order  derivatives. From Corollary \ref{iulie33} we know that the magnetic matrix  $\left\langle
\mathcal{W}^\epsilon_{\alpha,l},H^\epsilon\mathcal{W}^\epsilon
_{\beta,m}\right\rangle_{L^2(\X)}$ is (up to an error of order $\epsilon$ in the norm topology) unitarily equivalent with:
\begin{equation}\label{iulie15}
\exp\left\{-i\int_{[\alpha,\beta]}A_\epsilon\right\}\big(\widehat{\mu_{lm}}
\big)_{\alpha-\beta}.
\end{equation}
We compute:
\begin{align*}
\int_{[\alpha,\beta]}A_\epsilon&=\underset{1\leq j\leq d}{\sum}(\beta-\alpha)_j\int_{-1/2}^{1/2}dt\,A_j
\big(\epsilon(\alpha+\beta)/2+t\epsilon(\beta-\alpha)\big)\\
&=\left\langle A_\epsilon\big((\alpha+\beta)/2\big),\beta-\alpha\right\rangle\,+\,\epsilon\underset{j,k}{\sum}(\beta-\alpha)_j(\beta-\alpha)_k
\int_{0}^{1}ds\int_{-1/2}^{1/2}tdt\big(\partial_kA_j\big)\big(\epsilon st(\beta-\alpha)\big).
\end{align*}
Taking into account the hypothesis on the derivatives of the vector potential we obtain:
$$\exp\left\{-i\int_{[\alpha,\beta]}A_\epsilon\right\}-\exp\left\{-i\left\langle A_\epsilon\big((\alpha+\beta)/2\big),\beta-\alpha\right\rangle\right\}=\mathcal{O}\left (\epsilon <\alpha-\beta>^2\right).$$

Using the rapid decay of $\big(\widehat{\mu_{lm}}
\big)_\gamma$ with respect to $\gamma\in\Gamma$, we conclude that the spectrum of the matrix in \eqref{iulie15} is at a Hausdorff distance of order $\epsilon$ from the spectrum of  
\begin{equation}\label{iulie16} 
\exp\left\{-i\left\langle A_\epsilon((\alpha+\beta)/2),\beta-\alpha\right\rangle\right\}|E_*|^{-1}\int_{E_*}
\exp\{i\langle \theta,\alpha-\beta\rangle\}\mu_{lm}(\theta)d\theta \quad {\rm in}\quad l^2(\Gamma)^N.
\end{equation}
The above matrix can be identified with a bounded operator in $L^2(\X)^N\sim [l^2(\Gamma)\otimes L^2(E)]^N$ by taking the tensor product with the identity operator in $L^2(E)$. Its $N\times N$ matrix-valued integral kernel is then given by: 
\begin{equation}\label{iulie17} 
K^\epsilon([x]+\hat{x},[x']+\hat{x'}):=\delta(\hat{x}-\hat{x'})
\exp\left\{-i\left\langle A_\epsilon(([x]+[x'])/2),[x']-[x]\right\rangle\right\}\widehat{\mu}([x]-[x']).
\end{equation}
This operator is isospectral with the matrix in \eqref{iulie16}, hence its spectrum lies at a Hausdorff distance of order $\epsilon$ from the spectrum of $H^\epsilon E_I(H^\epsilon)$. 

Next we compute the integral kernel of the operator $\mathfrak{Op}(\mu^\epsilon)$ where $\mu^\epsilon(x,\xi)=\tilde{\mu}(\xi-A_\epsilon(x))$, and we shall compare it with $K^\epsilon$. The symbol $\mu^\epsilon(x,\xi)$ is $\Gamma_*$-periodic in the $\xi$ variable. 
The $N\times N$ matrix valued integral kernel of its Weyl quantization acting on $L^2(\X)^N$ is given by:
\begin{align}\label{novem1}
C^\epsilon([x]+\hat{x},[x']+\hat{x'}):=(2\pi)^{-d}\int_{\X^*} \mu^\epsilon((x+x')/2,\xi) e^{i\langle \xi,x-x'\rangle}d\xi.
\end{align}
Each $\xi\in \X^*$ can be uniquely written as $\xi=[\xi]+\theta$ with $[\xi]\in \Gamma_*$ and $\theta\in E_*$. We have $\mu^{\epsilon}(x,\xi)=\mu^\epsilon(x,\theta)$ and using $e^{i\langle \gamma^*,[x]-[x']\rangle}=1$ together with the completeness relation 
in $L^2(E)$
$$\sum_{\gamma^*\in \Gamma_*}\frac{(2\pi)^{d}}{|E_*|}
e^{i\langle \gamma^*,\hat{x}-\hat{x'}\rangle}=\delta(\hat{x}-\hat{x'})$$ we have:
\begin{align}\label{novem2}
C^\epsilon(x,x')&=\int_{E_*}d\theta (2\pi)^{-d}\sum_{\gamma^*\in \Gamma_*}
e^{ i\langle \gamma^*,x-x'\rangle }e^{ i\langle \theta,\hat{x}+[x]-\hat{x'}-[x']\rangle }\mu^\epsilon((\hat{x}+\hat{x'})/2+([x]+[x'])/2,\theta)\nonumber \\
&=\delta(\hat{x}-\hat{x'})|E_*|^{-1}\int_{E^*}d\theta e^{ i\langle \theta,[x]-[x']\rangle}
\widetilde{\mu}(\theta-A_\epsilon(\hat{x}+([x]+[x'])/2))\nonumber \\
&=\delta(\hat{x}-\hat{x'})
\exp\left\{-i\left\langle A_\epsilon(\hat{x}+([x]+[x'])/2),[x']-[x]\right\rangle\right\}\widehat{\mu}_{[x]-[x']}.
\end{align}
Because $E$ is bounded and $A$ has all its first order derivatives bounded we obtain the estimate:
$$\left\langle A_\epsilon(\hat{x}+([x]+[x'])/2),[x']-[x]\right\rangle-\left\langle A_\epsilon(([x]+[x'])/2),[x']-[x]\right\rangle=\mathcal{O}(\epsilon <[x']-[x]>).$$
Due to the localization properties of $\widehat{\mu}_\gamma$, we can replace  $A_\epsilon(\hat{x}+([x]+[x'])/2)$ with $A_\epsilon(([x]+[x'])/2)$ in \eqref{novem2} with the price of an error of order $\epsilon$ in the norm topology. But then we obtain the kernel in \eqref{iulie15} and the proof is over.

\vspace{0.5cm}

\noindent {\bf Acknowledgements}. Horia D. Cornean was supported by Grant 11-106598 of the Danish Council for Independent Research | Natural Sciences, and a Bitdefender Invited Professor Scholarship with IMAR, Bucharest. Radu Purice acknowledges the partial support of a grant of the Romanian National Authority for Scientific Research, CNCS-UEFISCDI, project number PN-II-ID-PCE-2011-3-0131 and thanks the University of Aalborg for its kind hospitality during part of the elaboration of this work.

\end{document}